\documentclass{article}
\usepackage[preprint]{cpal_2025}

\usepackage[utf8]{inputenc}
\usepackage{amssymb}
\usepackage{amsmath}

\usepackage{multirow}
\usepackage{graphicx}
\usepackage{enumitem}
\usepackage{amsfonts}
\usepackage{amsthm}
\usepackage{setspace}
\usepackage{tikz}
\usepackage{tikz-3dplot}
\usepackage{pgfplots}
\usepackage[normalem]{ulem} 
\usepackage{geometry} 
\usepackage[pageanchor=false,colorlinks=true,linkcolor=blue,citecolor=blue,urlcolor=blue]{hyperref}
\usepackage{cleveref}

\usepackage{subcaption}
\usepackage{booktabs}

\usetikzlibrary{calc,intersections,positioning,angles,quotes}
\usetikzlibrary{3d,shapes}
\pgfplotsset{compat=1.11}

\newcommand\bR{\mathbb{R}}

\newcommand{\supp}{\mathtt{supp}}

\newcommand{\gt}{\hat{x}}

\newtheorem{theorem}{Theorem}[section]
\newtheorem{corollary}{Corollary}[section]

\newtheorem{lemma}{Lemma}[section]
\newtheorem{definition}{Definition}[section]

\begin{document}
	
	\title{From sparse recovery to plug-and-play priors, understanding trade-offs for stable recovery with generalized projected gradient descent}
	\author{
		Ali Joundi\thanks{Univ. Bordeaux, Bordeaux INP, CNRS, IMB, F-33400, Talence, France
			\texttt{\{ali.joundi,yann.traonmilin, jean-francois.aujol\}@math.u-bordeaux.fr}.}
		,~ Yann  Traonmilin\footnotemark[1],~ Jean-François Aujol\footnotemark[1] }

	\maketitle
	
	\begin{abstract}
		We consider the problem of recovering an unknown low-dimensional vector from noisy, underdetermined observations. We focus on the Generalized Projected Gradient Descent (GPGD) framework, which unifies traditional sparse recovery methods and modern approaches using learned deep projective priors. 
		We extend previous convergence results to robustness to model and projection errors. We use these theoretical results to explore ways to better control stability and robustness constants. To reduce recovery errors due to measurement noise, we consider generalized back-projection strategies to adapt GPGD to structured noise, such as sparse outliers. To improve the stability of GPGD, we propose a normalized idempotent regularization for the learning of deep projective priors. 
		We provide numerical experiments in the context of sparse recovery and image inverse problems, highlighting the trade-offs between identifiability and stability that can be achieved with such methods. 
	\end{abstract}
	\section{Introduction}
	We study the problem of recovering an unknown vector $\gt \in \bR^N$ from an undetermined number of noisy observations $y \in \bR^m$ defined by:
	\begin{equation}\label{eq:obs}
		y = A\gt+e
	\end{equation}
	where $A$ is the measurement operator and $e \in \bR^m$ is the measurement noise. Many problems in data science can be modeled this way, in particular, imaging problems, where $\gt$ is an image (in biology, medicine, astronomy,  \dots).
	
	As the number of measurements is often insufficient, $ m < N$, a prior model on the unknown is \emph{necessary}. In this article, we assume $\gt$ lies approximately in $\Sigma$, where $\Sigma$ is a \emph{low-dimensional} model, i.e., a subset of $\bR^N$ that can be described with few parameters. This setup is the basis of sparse recovery theory, where, under some conditions on the measurement operator $A$ (e.g., number of measurements of random Gaussian matrices), it is possible to guarantee stable (with respect to noise) and robust recovery (with respect to model error) of elements of $\Sigma$ with convex or non-convex algorithms. For example, for noise of finite energy, given an estimate $x^\star$ of $\gt$, such guarantees are expressed,
	\begin{equation}\label{eq:guarantees}
		\|x^\star - \gt \|_2^2 \leq C_1 \|e\|_2^2 + C_2 d(\hat{x},\Sigma)
	\end{equation}
	where $C_1,C_2$ are constants and $d(\hat{x},\Sigma)$ is some notion of distance to $\Sigma$.
	
	Learning based methods, where the prior on the unknown vector $\hat{x}$ is learned, have been very successful at solving such inverse problems when a large database $X \subset \bR^N$ of examples is available. A large part of the literature uses the minimization of potentially non-convex functionals to solve such problems. In particular, the popular plug-and-play methods use a general-purpose denoiser learned on $X$ as a projective prior used for the minimization of such functionals and consider convergence to critical points of such functions.  Many learning based methods can be understood as methods using \emph{deep projective priors} where a generalized projection onto a non-explicit set $\Sigma$ is learned. In \cite{traonmilin2024towards}, instead of considering the minimization of an underlying functional, it was proposed to consider unified guarantees of generalized projected gradient descent (GPGD). It was shown that GPGD for low-dimensional recovery can model, at the same time  sparse recovery (through the classical iterative hard thresholding algorithm) and a class of learning based plug-and-play methods where the denoiser is used as a projection. We define GPGD iteration as
	\begin{equation}\label{eq:def_PGD_intro}
		\mathcal{I}(P_\Sigma):\quad\quad x_{n+1} = P_{\Sigma}(x_n) - \mu A^T(AP_{\Sigma}(x_n)-y) 
	\end{equation}
	where $P_{\Sigma}$ is a generalized projection on the low dimensional model $\Sigma$ and $\mu>0$ is a fixed step size. Note that projected gradient descent is often presented with the projection and descent step reversed.  The "gradient" in \eqref{eq:def_PGD_intro} is the gradient of the $\ell^2$-datafit $\nabla\frac{1}{2}\|Ax-y\|_2^2 =A^T(Ax -y)$. This gradient can also be interpreted as a back-projection of the residual $Ax-y$ to the ambient space of $\hat{x}$ through the back-projection operator $A^T$. 
	
	For sparse recovery, using the orthogonal projection leads to the iterative hard thresholding algorithm. In the context of plug-and-play methods, using a general-purpose denoiser as a projection leads to the proximal gradient descent plug-and-play method. Other methods relying on deep projective priors, such as auto-encoders, can be interpreted in this framework~\cite{joundi2025stochastic}. \cite{traonmilin2024towards} shows that global linear recovery is possible provided the measurement operator $A$ verifies a restricted isometry property and the projection verifies a \emph{restricted Lipschitz property}. It is further shown that the restricted Lipschitz constant drives both the identifiability capabilities of $\Sigma$ and the convergence rate. It is advocated that this constant is thus a good property to compare different GPGD algorithms to recover the same model and even to consider optimal methods (in this context, iterative hard thresholding is shown to be near-optimal for sparse recovery).
	
	However, \cite{traonmilin2024towards} only considers a perfect modeling without noise, model error and with "ideal" projection having the restricted Lipschitz property. \cite{joundi2025stochastic} showed that stable recovery (stability to observation noise) was guaranteed under the same conditions. However, stability to noise has not been fully explored. For example, the considered PGD algorithm uses the gradient $A^T(Ax-y)$ of the $\ell^2$ data-fit functional $\frac{1}{2}\|A x - y\|_2^2$. While this gradient is well adapted to Gaussian white noise (as it is related to the maximum likelihood estimator), it might not be adapted to other types of degradations such as sparse corruptions, where the energy of the noise cannot be bounded. The validity of initial results with model error and approximate projections is still an open question. 
	
	In this paper, we study how the generalized projected gradient descent behaves under the presence of noise, modeling error, and approximate projection, in a framework that unifies sparse recovery and methods relying on deep projective priors. In particular, we study whether the restricted isometry and restricted Lipschitz conditions are sufficient to guarantee stable and robust recovery. We also discuss how to minimize or control different error terms in the design of GPGD methods.
	
	\subsection{Contributions}
	
	In Section~\ref{sec:theory}, we provide a general stable and robust recovery theorem for generalized projected gradient descent with arbitrary back-projections. This theorem shows how using GPGD with general back-projections and approximate projections leads to stability to measurement noise, modeling error, and approximate projections. In such a context, the restricted isometry constants and restricted Lipschitz conditions are still the main factors impacting recovery.
	
	In Section~\ref{sec:noise_adapt}, when sparse outliers contaminate measurements, we illustrate how adapting the back-projection leads to stable recovery and provide experiments within the plug-and-play framework. 
	
	In Section~\ref{sec:idempotent}, we propose a normalized idempotent regularization (NIPR) to control the approximate projection error. We show experimentally that NIPR for deep projective priors (plug-and-play priors and auto-encoder priors in Section~\ref{sec:additional_exp} in the appendix) improves the stability of convergence of GPGD while preserving reconstruction quality.

	\subsection{Related work}

	There exists a wide body of work studying the projected gradient descent in various contexts. For sparse recovery, PGD is known as iterative hard thresholding, and its linear convergence under a restricted isometry property of the operator $A$ has been shown in \cite{blumensath2010normalized,foucart2011hard}. Similar results for function minimization with PGD without an explicit low-dimensional model were given in \cite{oymak2017sharp}. Approximate \emph{orthogonal} projections for the recovery of low-dimensional models are studied in \cite{golbabaee2018inexact}. In \cite{bahmani2016learning}, global convergence of PGD is given for a class of generalized sparsity models and the orthogonal projection. \cite{traonmilin2024towards} decouples the convergence rate through a restricted isometry constant and a restricted Lpischitz constant of the projection. In this work, we will consider approximate \emph{restricted Lipschitz} projections and stable recovery for generic noises. Global convergence of gradient projection has been shown under a general KL property in \cite{attouch2013convergence}. General stationary properties of the iterates of PGD are given in~\cite{olikier2024projected}. In this work, we only consider cases where linear convergence is proven.
	
	Beyond the optimality approach in \cite{traonmilin2024towards}, other works intend to formally improve PGD algorithms. For thresholding algorithms, \cite{liu2020between} optimizes a local concavity property to improve local convergence properties. In \cite{kamilov2016learning}, an optimal non-linearity is learned from the data. In \cite{ollila2014robust}, a robust iterative hard thresholding algorithm is presented: it proposes to use $A^TW$ where $W$ is a shrinkage operator calculated iteratively. However, no theoretical analysis is provided. In~\cite{bhatia2015robust}, an iterative thresholding applied solely to the residual is presented (no sparsity model on the data). In this work, we will discuss adaptation to noise in the context of outliers and the impact on recovery guarantees. Note that adaptation to noise in the variational context (minimization of a regularized data-fit functional) is generally done by adapting the norm of the data-fit (see e.g. \cite{popilka2007signal,studer2011recovery,traonmilin2015robust} for sparse outliers).
	
	The generalized projected gradient descent context allows us to study recovery algorithms relying on deep priors. Specifically, in the class of state-of-the-art plug-and-play (imaging) algorithms \cite{venkatakrishnan2013plug,cohen2021regularization,chen2021deep,kamilov2023plug}, the so-called proximal gradient method \cite{zhang2021plug} is a generalized projected gradient descent algorithm where the projection is performed using a general-purpose denoiser learned with a deep neural network. While \cite{traonmilin2024towards,joundi2025stochastic} showed experimentally that this method linearly converges to approximate fixed points of the denoiser, a complete theoretical study of stability, robustness, and approximation of a projection with deep neural networks was not given.

	In \cite{shocher2023idempotent}, in the context of generative modeling, the importance of having access to a real projection onto the data manifold is emphasized. In consequence, an idempotent regularization of the generative model is proposed (i.e $ \|f\circ f-f\|$ where $f$ is the generating function) and shows improved stability of the generation. We will build on this idea to improve the stability of deep projective priors. Note that there exists a recent wide literature on stochastic algorithms using generative models such as diffusion or flow matching priors to solve inverse problems (see e.g. surveys \cite{he2025diffusion, daras2024survey}). While links with deep projective priors have been made in deterministic settings,  such algorithms are out of the scope of this article where we focus on deterministic "plug-and-play like" algorithms.

	\section{Notations, definitions and previous results}\label{sec:prev_results}
	We introduce our theoretical framework and recall previous results useful to understand our main theoretical result in the next section.  We suppose that $\Sigma$ is a homogeneous space (verified by sparse and low-rank models).
	
	To guarantee uniform recovery, we assume that the measurement operator $A$ verifies a restricted isometry property (a lower RIP is anyway \emph{necessary} for the identifiability of $\Sigma$ \cite{bourrier2014fundamental}). We use the following notion of restricted isometry constant.
	
	\begin{definition}\label{def:RIC}
		The operator $B$ has a restricted isometry constant (RIC) $\delta <1$ on the secant set $\Sigma-\Sigma =\{x_1-x_2 : x_1,x_2 \in \Sigma \}$ if for all $x_1,x_2 \in \Sigma$
		
		\begin{equation}
			\|(B-I)(x_1-x_2)\|_2\leq \delta \|x_1-x_2\|_2
		\end{equation}
		We write $\delta_\Sigma(B)$ the smallest admissible restricted isometry constant (RIC).
	\end{definition}
	
	Definition~\ref{def:RIC} is well adapted to the study of the composition of the measurement operator with a generalized back-projection of the form $LA$ (where $L$ is a linear operator) through the constant $\delta(LA)$ (in iterations~\eqref{eq:def_PGD_intro}, $L=A^T$). We consider the following notion of generalized projection.
	
	\begin{definition}[Generalized projection]
		Let $\Sigma \subset \bR^N$. A (set-valued) generalized projection onto $\Sigma$ is a function $P$ such that for any $z\in\bR^N$, $P(z) \subset \Sigma$.
	\end{definition}
	
	By abuse of notation, to facilitate reading, an equation true for any $w \in P(z)$ is written using the notation $P(z)$. Considering such generalized projection brings flexibility to the framework: both sparse models and deep projective priors can be considered (as the latter might not be projections in the usual mathematical definition as idempotent operators $P\circ P = P$). 
	
	\begin{definition}[Orthogonal projection]\label{def:orth_proj}
		We define, when it exists, the (set-valued) orthogonal projection onto a set $\Sigma \subset \bR^N$ as follows: for all $z\in\bR^N$
		\begin{equation}
			P_\Sigma^\perp(z) = \arg \min_{x\in\Sigma} \| x-z \|_2^2.
		\end{equation}
	\end{definition}
	
	We recall the restricted Lipschitz property of the projection introduced by \cite{traonmilin2024towards} to study linear convergence of GPGD.
	
	\begin{definition}[Restricted Lipschitz property]\label{def:lip_const}
		Let $P : \bR^N \to \bR^N$. Then $P$ has the restricted $\beta$-Lipschitz property with respect to $\Sigma$ if for all $z \in \bR^N, x \in \Sigma, u \in P(z)$ we have
		\begin{equation}
			\begin{split}
				\|u-x\|_2 &\leq \beta \|z-x\|_2\\
			\end{split}
		\end{equation}
		We denote by $\beta_{\Sigma}(P)$ the smallest $\beta$ for which $P$ satisfies the restricted $\beta$-Lipschitz property.
	\end{definition}
	Note that for single-valued projection, this can be rewritten $\|P(z)-P(x)\|_2 \leq \beta \|z-x\|_2$, hence the name \emph{restricted Lipschitz}. Linear recovery of $\hat{x}$ with GPGD has been shown when $\delta(A^TA) \beta_\Sigma(P) <1$.
	%
	If the orthogonal projection $P_\Sigma^\perp$ exists ($\Sigma$ is then called \emph{proximinal}), we have $1 \leq \beta_\Sigma(P_\Sigma^\perp)\leq 2$. It was shown that the optimal projection $P^\star$ minimizing $\beta_\Sigma(P^\star)$ exists. For sparse recovery ($\Sigma= \Sigma_k$, the set of sparse vectors), using the orthogonal projection corresponds to the iterative hard thresholding algorithm and is nearly optimal for the restricted Lipschitz constant with $\beta_{\Sigma_k}(P_{\Sigma_k}^\perp)\leq \sqrt{\frac{3 +\sqrt{5}}{2}}  \approx 1.618$. Also note that if $\beta_\Sigma(P) < \infty$ then $P$ is necessarily idempotent.

	\section{Stable and robust linear recovery with GPGD with approximate projections} \label{sec:theory}
	
	We give a general recovery result with GPGD with generalized back-projections and approximate projections. We include three deviations from an ideal noiseless model: stable recovery with generalized back projection, robustness to model error and robustness to approximate projections. 
	
	\noindent \textbf{Adaptation to noise} As was proposed in previous work~\cite{traonmilin2024towards}, the restricted Lipschitz constant is a good objective to minimize in the search for optimal projections for GPGD algorithms given a low-dimensional model $\Sigma$. However, it is not clear in a noisy context what would be an optimal GPGD as we have guarantees of the form:
	\begin{equation}
		\|x_n - \gt\|_2^2 \leq C_1 r^n + C_2\|e\|_2
	\end{equation}
	where $x_n$ are the iterates of the considered GPGD method and $C_1,C_2,r \geq 0 $ are some constants.
	
	To define an optimal GPGD method as minimizing such an upper bound, we need to optimize both the rate $r$ and the stability constant $C_2$. When there is some knowledge about the structure of the noise, variational methods provide a way to adapt the recovery method by adapting the norm of the data fit term. In particular, it is known that in some cases, stability to unbounded noise (sparse corruptions) can be obtained using a sparsity-inducing norm such as the $ \ell^1$-norm for the data-fit term \cite{popilka2007signal,studer2011recovery,traonmilin2015robust}.  Following similar ideas, we propose to generalize the projected gradient descent by considering general back-projections instead of the back-projection $A^T$ induced by a $\ell^2$ data-fit. 
	
	\noindent \textbf{Model error} In practice, we have that $\hat{x} \notin \Sigma$, but we assume some control $ d_{P_\Sigma}(\hat{x},\Sigma) \leq \tau$ (or equivalently $\|\hat{x}- P_\Sigma(\hat{x})\|_2 \leq \tau$) for the distance to the model $d_{P_\Sigma}(\cdot,\Sigma) $ associated to the projection $P_\Sigma$.

	\noindent \textbf{Approximate projections} In the plug-and-play framework, in GPGD, we use a general-purpose denoiser as $P_\Sigma$. While ensuring that $P_\Sigma$ is actually a projection (i.e., it has a set of fixed points) would guarantee linear convergence, in practice, it is not a real projection. In fact, it is often observed that iterations of PGD in this context diverge after having reached an optimal point. To model this effect, suppose that in place of a restricted $\beta$-Lipschitz projection $P_\Sigma$, we use a projection $P$ such that:
	\begin{equation}
		P(x_n) = P_\Sigma(x_n) + R(x_n)
	\end{equation}
	where $\|R(x_n)\|_2 \leq \eta$ for some constant $\eta>0$.

	We consider GPGD with generalized back-projection iterations:
	\begin{equation} \label{def:gen_PGD}
		\begin{split}
			\mathcal{I}_{GBP}(P,L): \quad \quad x_{n+1} &= P(x_n) - \mu L(AP(x_n)-y)
		\end{split}
	\end{equation}
	where $L \in \bR^{N\times m} $ is a general linear back-projection from the observation space to the ambient space of $\gt$ and $P_\Sigma$ is a generalized projection onto $\Sigma$. This way, the complexity of the back-projection step is limited to the cost of a matrix vector multiplication. We present our main convergence result.
	\begin{theorem} \label{th:stability}
		Let $\Sigma \subset \bR^N$. Let $\mu,\eta >0$. Let $P_\Sigma$ be a generalized projection onto $\Sigma$. Consider iterates from the GPGD with approximate projection $P = P_\Sigma+R$ and $\|R(x)\|_2\leq \eta$ for all $x\in\bR^N$. Suppose that $\mu L A$ has restricted isometry constant $\delta := \delta(\mu L A)$ and $P_\Sigma$ has restricted Lipschitz constant $\beta := \beta_\Sigma(P_\Sigma)$. Consider $d_{P_\Sigma}(\gt,\Sigma) := \|P_\Sigma(\gt)-\gt\|_2$. Assuming $\delta\beta<1$, we have
		\begin{equation}
			\begin{split}
				\|x_{n} -P_\Sigma(\gt)\|_2& \leq (\delta \beta)^{n}\|x_0-P_\Sigma(\gt)\|_2 + C_{\mathrm{stab}} \| L e  \|_2 +C_{\mathrm{rob}} d_{P_\Sigma}(\gt,\Sigma) + C_{\mathrm{proj}} \eta
			\end{split}
		\end{equation}
		where we define the stability constant $C_{\mathrm{stab}} := \frac{\mu}{1-\delta \beta} $;
		the robustness (to model error) constant $C_{\mathrm{rob}} := \frac{1}{1 -\delta\beta}\| \mu L A\|_{\mathrm{op}}$;
		the approximate projection constant $C_{\mathrm{proj}} := \frac{1}{1 -\delta\beta}\| I-\mu L A\|_{\mathrm{op}}$.
		We also have
		\begin{equation}
			\begin{split}
				\|x_{n} -\gt\|_2& \leq  (\delta \beta)^{n}\|x_0-P_\Sigma(\gt)\|_2 + C_{\mathrm{stab}}  \| \mu L e  \|_2 + C_{\mathrm{rob}}' d_{P_\Sigma}(\gt,\Sigma) + C_{\mathrm{proj}} \eta
			\end{split}
		\end{equation}
		where $C_{\mathrm{rob}}' := 1+\frac{\|\mu L A\|_{\mathrm{op}}}{1-\delta \beta}$.
	\end{theorem}
	
	This theorem shows linear convergence to a set of estimates having estimation error controlled by the noise level, the model error and the approximation of the projection with the constants $C_{\mathrm{stab}}$, $C_{\mathrm{rob}}$ and $C_{\mathrm{proj}}$, respectively. Particularly, in the specific case when $L=A^T$, $d_{P_\Sigma}(\gt,\Sigma)=0$, $\eta=0$, we obtain recovery guarantees that were given in \cite{traonmilin2024towards, joundi2025stochastic}. 
	
	With Theorem~\ref{th:stability}, we remark that for a given measurement operator $A$, the restricted Lipschitz constant of the projection $P_\Sigma$ drives both the rate of convergence and identifiability of  GPGD with backprojection $L \propto A^T$. The stability constant $C_{\mathrm{stab}}$, the robustness constant $C_{\mathrm{rob}}$ and the approximate projection constant $C_{\mathrm{proj}}$  are  also increasing with respect to $\beta_\Sigma(P_\Sigma)$. We conclude that choosing $P_\Sigma$ minimizing $\beta_\Sigma(P)$ as was proposed in the noiseless case in \cite{traonmilin2024towards} is a reasonable notion of optimal projection within the class of algorithms $\{\mathcal{I}_{GBP}(P_\Sigma,L)\}$ independently of the backprojection $L$. We illustrate the importance of the restricted Lipschitz condition in the case of sparse recovery in Section~\ref{sec:exp_lip_approx} in the appendix. We also immediately remark that if we want to adapt the backprojection to improve stability to noise in $\{\mathcal{I}_{GBP}(P_\Sigma,L)\}$, we will face trade-offs in terms of convergence rate and identifiability through the RIC $\delta(\mu LA)$. Consequently, in Section~\ref{sec:noise_adapt}, we discuss how the generalized back-projection $L$ can be adapted to control the stability constant. In Section~\ref{sec:idempotent}, we show how the approximate projection error can be controlled in the training of deep projective priors with idempotent regularization. 
	
	\section{Trade-offs for stable linear recovery with GPGD with generalized backprojections} \label{sec:noise_adapt}
	
	In this section, to simplify the exposition and focus on the adaptation of GPGD to noise, we suppose that $\gt \in \Sigma$ and that $P_\Sigma$ is a restricted $\beta$-Lipschitz projection onto $\Sigma$. We discuss how the back-projection can be adapted to structured sparse noise (outliers) through the generalized back-projection $L$. Suppose $e$ is a $s$-sparse noise with unbounded energy (i.e. $\|e\|_2$ is very large). This can model saturation noise, occlusions, or dead samples in signal and image processing, where the amplitude information is completely lost or if the sensor is saturated (e.g., in very bright light conditions). In this case, we cannot hope to acceptably bound $\|Le\|_2$ if $L$ is full rank (we might be able to trade off some convergence speed for improved stability constant, see Section~\ref{sec:impact_mu} in the Appendix). However, in these cases, we can often estimate the support of the noise. For saturated pixels, we just need to select pixels equal to $1$ for images coded in $[0,1]^N$. Then, for any $L = BS$ where $S$ is a diagonal matrix in $\{ 0,1\}^{m \times m}$ that selects the complement of the support of $e$ (precisely, $S = \mathrm{diag}( \mathrm{1}_{\supp(e)^c})$) and $B \in \bR^{N\times m}$, we have
	\begin{equation}
		\|Le\|_2 = \|BSe\|_2 = 0.
	\end{equation}

	Take e.g. $L = A^TS$, we get that $\delta(LA)= \delta(A^TSA)$ and $\|Le\|_2 = \|A^TSe\|_2 = \|0\|_2 = 0 $, which gives
	\begin{equation}
		\begin{split}
			\|x_{n} -\gt\|_2& \leq (\delta(A^TSA) \beta)^{n}\|x_0-\gt\|_2
		\end{split}
	\end{equation}
	
	Thus, stable linear recovery is achieved if the operator $A^TSA$ has  RIC with  $\delta(A^TSA) \beta< 1$. This can be seen as a trade-off between stability to noise and identifiability and convergence speed as the thresholding operation removes measurements. We illustrate these results for deep projective priors (and sparse recovery in Annex~\ref{sec:tradeoff_sparse}).
	
	In Figure \ref{fig:outliers}, we show that robustness to outliers is achieved by adapting the back-projection. We use the plug-and-play approach (the approximate projection $P$ is a learned denoiser, see next sections for more details) for solving a super-resolution inverse problem for CelebA images ($A$ is a subsampling by a factor $2$). In particular, Figure \ref{fig:outliers_graph} represents the evolution of the normalized error with respect to the number of outliers $s$. For different values of the number of dead samples $s$ as represented in \ref{fig:noisy_outliers}, we solve a super-resolution inverse problem with  GPGD  with back-projection $A^TS$. We observe that when GPGD is adapted to this particular noise structure, it is possible to obtain a robust estimation of the ground truth which is not the case without. In addition to that, graph \ref{fig:outliers_graph} shows that when the number of outliers increases beyond $2000$ we observe a phase transition where robustness is no longer observed, a phenomenon well predicted by our theoretical findings and also observed in the case of classical sparse recovery (see Section~\ref{sec:tradeoff_sparse} in the appendix).
	\begin{figure}[h!]
		\centering
		\begin{subfigure}[t]{0.20\textwidth}
			\centering
			\includegraphics[width=\textwidth]{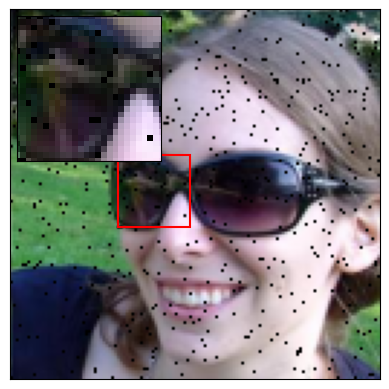}
			\subcaption{Noisy image with 2000 dead pixels and a sub-sampling of factor 2}\label{fig:noisy_outliers}
		\end{subfigure}
		\hspace{0.1cm}
		\begin{subfigure}[t]{0.20\textwidth}
			\centering
			\includegraphics[width=\textwidth]{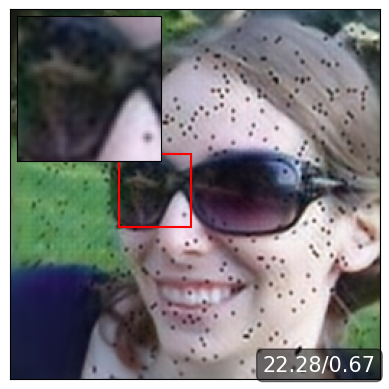}
			\subcaption{GPGD recovery without adapted back-projection}\label{fig:pgd_outliers}
		\end{subfigure}
		\hspace{0.02cm}
		\begin{subfigure}[t]{0.20\textwidth}
			\centering
			\includegraphics[width=\textwidth]{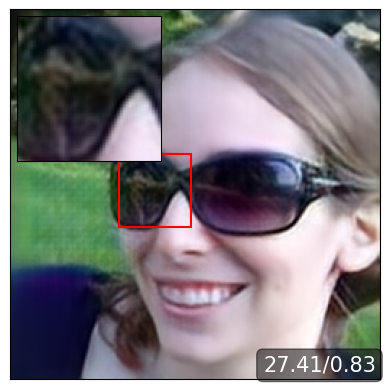}
			\subcaption{GPGD recovery with adapted back-projection}\label{fig:pgd_outliers}
		\end{subfigure}
		\hspace{0.02cm}
		\begin{subfigure}[t]{0.35\textwidth}
			\centering
			\includegraphics[width=\textwidth]{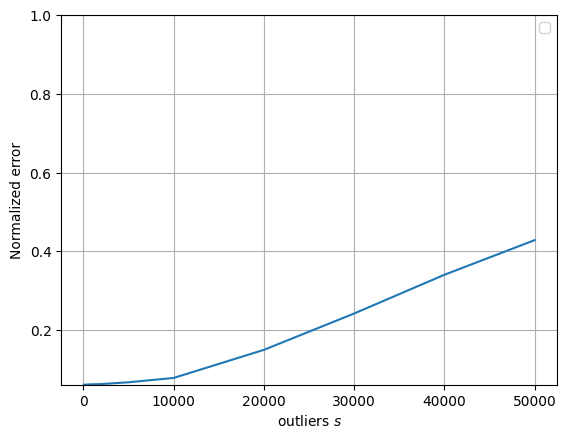}
			\subcaption{Evolution of the normalized error for super-resolution $\times 2$ inverse problem with the number of outliers $s$. }\label{fig:outliers_graph}
		\end{subfigure}

		\caption{Adaptation to sparse noise: Illustration of the trade-off between stability to noise sparsity and identifiability of $\Sigma$. We show the Normalized error bound for 90$\%$ of the experiments with respect to the sparsity $s$ of outliers. The sparser the noise is, the greater the identifiability.}\label{fig:outliers}
	\end{figure}
	
	\section{Mitigating the effect of approximate projections with normalized idempotent regularization}\label{sec:idempotent}
	
	When training a projective prior $P_\Sigma$ into the form of an auto-encoder or a denoiser (plug-and-play), we showed that it is sufficient to control the approximate projection constant $C_{\mathrm{proj}}$ to guarantee approximate stable and robust recovery. Given $X$ a database, deep projective priors can be trained with the following classical loss functions: 
	\begin{equation}\label{eq:def_loss}
		\begin{split}
			\mathcal{L}_{X,\text{AE}}(P) &=\mathcal{L}_{X,\text{AE}}(f_D\circ f_E):=\sum_{x\in X} \|f_D\circ f_E(x)-x\|_2^2,~\text{for an autoencoder},\\
			~ \mathcal{L}_{X,\text{PnP}}(P) &=\mathcal{L}_{X,\text{PnP}}(D):=\mathbb{E}_{x\in X, \varepsilon\sim\mathcal{N}(0,\xi^2\mathbf{I})}\left(\|D(x+\varepsilon)-x\|_2^2\right),~\text{for a denoiser}.
		\end{split}
	\end{equation} 
	With these loss functions, it is typically observed that $P \circ P (x) \neq P(x)$. The idempotent property that $P \circ P = P$ is a necessary condition for $P$ to have the restricted Lipschitz property~\cite{traonmilin2024towards}. For some projective priors, this problem often leads to instabilities near convergence. 
	
	In \cite{joundi2025stochastic}, it was proposed to use a stochastic orthogonal regularization (SOR) for the training of deep projective priors to control the restricted Lipschitz constant by trying to learn an approximate orthogonal projection. SOR showed improved convergence speed and identifiability properties for very ill-posed inverse problems. However, while efficient with an oracle stopping criterion, the corresponding regularized projective priors still show some instabilities near convergence of GPGD (as shown in experiments). We propose to explore the effect of idempotent regularization of deep projective priors for the stability of GPGD. Note that this regularization has been proposed in the context of generative modeling in \cite{shocher2023idempotent} to interpret the generative process as a projection. Our objective is to approximate an idempotent projection during training to better control the projection error term in Theorem~\ref{th:stability}. We define the Normalized Idempotent Regularization (NIPR) as:
	\begin{equation}
		\mathcal{R}_X(P) = \sum_{x\in X} \frac{\|P\circ P(x)-P(x)\|_2}{\|P(x)\|_2},
	\end{equation}
	and the regularized loss functions
	\begin{equation}\label{eq:def_loss_reg}
		\begin{split}
			\mathcal{L}_{X,\text{AE}}^{\mathrm{reg}}(P) &=	\mathcal{L}_{X,\text{AE}}(P) + \lambda 	\mathcal{R}_X(P);\\
			\mathcal{L}_{X,\text{PnP}}^{\mathrm{reg}}(P) &=\mathcal{L}_{X,\text{PnP}}(P) + \lambda 	\mathcal{R}_X(P).\\
		\end{split}
	\end{equation} 
	
	Note that compared to \cite{shocher2023idempotent}, we normalize the idempotent criterion to avoid a bias towards $P(x)$ of low energy. Indeed, given a linear generalized projection, consider the family of functions $\alpha P$ with $\alpha > 0$, we have that 
	\begin{equation}
		\|(\alpha P)\circ (\alpha P)(x)- \alpha P(x)\|_2 = \alpha \|\alpha P\circ P(x)- P(x)\|_2 \to_{\alpha \to 0} 0. 
	\end{equation}
	Hence, this idempotent criterion without normalization can be made arbitrarily small by projections with small norm. 
	
	\textbf{Experiments} We now apply our NIPR to a deep denoising neural network on the CelebA dataset \cite{liu2015faceattributes} trained with a PnP loss. We train a DRUNET denoiser \cite{zhang2021plug}, a U-net combined with skip connections on a dataset of size 10000. The images represent RGB faces of size 256x256. We considered two denoisers: one without NIPR and another with NIPR (with $\lambda=0.005$). The NIPR is computed over the same batch of images as the one used from the denoising loss $\mathcal{L}_{X,\text{PnP}}$. After the training, we solve a super-resolution inverse problem by a factor of 2. We consider two image restoration algorithms: the classical unconstrained GPGD and GPGD using a regularized denoiser with NIPR. The main baseline here is PGD. Additionally, as a comparison reference, we compute the stability of the Stochastic Orthogonal Regularization (SOR \cite{joundi2025stochastic}) added to a third DRUNET, as done for NIPR.
	We consider for our experiments a test set of 50 images from CelebA. To compare the recovery, we use the Peak-Signal-to-Noise-Ratio (PSNR) and the Structural Similarity Index Measure (SSIM) of recovered images. The reader can find in the appendix additional experiments on inpainting and deblurring inverse problems. We also test our regularization over an autoencoder  DPP trained on the MNIST dataset. 
	
	To study the impact of NIPR, we define two metrics that assess it.  The goal in using these two stability metrics is to cover several instability cases \textit{e.g.} to assess whether the quantity $\|x-\hat{x}\|$ is diverging or oscillating after reaching the optimal value. 
	
	The Stability Metric 1 (SM1) is defined as \begin{equation}
		\text{SM1}(\hat{x},n)=\underset{i_{\min +1}\leq i\leq i_{\min}+n}{\max}\left(\frac{\|x_i-\hat{x}\|_2}{\|x_{\min}-\hat{x}\|_2}-1\right), 
	\end{equation}
	where $\hat{x}$ is the ground truth, $x_i$ the solution at iteration $i$, $x_{\min}$ the optimal one at $i_{\min}$ and $n$ the length of the interval where we want to compute the stability. This metric captures the total deviation given a number of iterations after reaching the optimal estimate.
	
	The  Stability Metric 2 (SM2) is defined as \begin{equation}
		\text{SM2}(\hat{x},n)=\sum_{i=i_{\min}+1}^{ i_{\min}+n}{\frac{\|x_{i+1}-x_i\|_2}{\|x_i\|_2}}.
	\end{equation}
	This metric captures oscillation phenomenon after reaching the optimal estimate.
	
	Note that generally $i_{\min}$ is an oracle minimizing $\|x_i-\hat{x}\|_2$ that is not available in a real situation. Hence, controlling stability metrics can help making stopping criteria for GPGD more stable.
	
	Table \ref{tab:supres_fig} represents the PSNR and the SSIM of the recovered images from a super-resolution inverse problem for the three aforementioned methods. The stability metrics SM1 and SM2 are averaged through the test set and are computed at different iterations: $n=\{i_{\min}+10, i_{\min}+50, i_{\min}+100\}$, where $i_{\min}$ is the iteration of the optimal solution. Figures \ref{fig:supres_pgd_graph_n_visu} and \ref{fig:supres_nipr_graph_n_visu} represent the evolution of the quantity $\|x-\hat{x}\|$ for all the tests for the considered method. The graph is accompanied by visual results at different iterations: $i_{\min}+10,~i_{\min}+50,~i_{\min}+100$.
	\begin{table}[h!]	
		\caption{PSNRs and stability values for a super-resolution inverse problem using GPGD without regularization, NIPR and SOR. While recovering the original image  correctly, NIPR is clearly more stable after reaching $x_{\min}$ compared to PGD and SOR.}\label{tab:supres_fig}
		\begin{minipage}{0.48\textwidth}
			\centering
			\begin{tabular}{ccccccccc}
				\toprule
				\multirow{2}{*}{Method} & \multirow{2}{*}{PSNR$\uparrow$} & \multirow{2}{*}{SSIM$\uparrow$} & \multicolumn{3}{c}{SM1 $\downarrow$}  & \multicolumn{3}{c}{SM2 $\downarrow$}  \\
				
				&            &            & $i_{\min}+10$ & $i_{\min}+50$ & $i_{\min}+100$ & $i_{\min}+10$ & $i_{\min}+50$ & $i_{\min}+100$ \\
				\hline
				No reg.           & \textbf{34,160}        & \textbf{0,933}         & 0,0510  & 3,9210  & 10,2517  & 0,018   & 0,257   & 0,799   \\
				NIPR          & 33,346        & 0,913         & \textbf{0,0109}  & \textbf{0,2487}  & \textbf{0,7732}   & \textbf{0,014}   & \textbf{0,045}   & \textbf{0,095}   \\
				SOR           & 32,855        & 0,913         & 0,0632  & 3,4256  & 18,1521  & 0,020   & 0,254   & 1,145 \\
				\hline   
			\end{tabular}
		\end{minipage}
	\end{table}
	According to table \ref{tab:supres_fig}, each method recovers the original images correctly. As it is more constrained, NIPR performs slightly worse than GPGD without regularization but still has acceptable performance. However, when comparing the stability metrics, we see clearly that NIPR has a significant impact. In fact, for both SM1 and SM2, deviations are limited for NIPR, whereas they increase significantly through the iterations for GPGD with no regularization and SOR. In particular, this result shows that NIPR has, in fact, guaranteed a more stable convergence, compared to the two other methods. It is interesting to note that even though SOR reaches the optimal solution faster, it becomes less stable compared to other methods in return. Visual results in figures \ref{fig:supres_pgd_graph_n_visu} and \ref{fig:supres_nipr_graph_n_visu} confirm our observations.  GPGD without regularisation diverges quickly after reaching $x_{\min}$, thus producing heavily degraded images. Conversely, NIPR still produces acceptable recoveries even if we are far from the optimal solution.
	
	\begin{figure}[h!]
		\centering
		\begin{subfigure}[t]{0.36\textwidth}
			\centering
			\includegraphics[width=\textwidth]{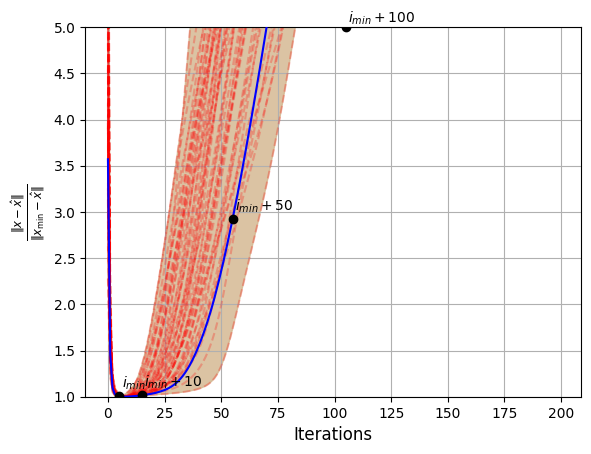}
			\subcaption{Evolution of $\frac{\|x_i-\hat{x}\|}{\|x_{\min}-\hat{x}\|}$ }
		\end{subfigure}
		\hspace{0.5cm}
		\begin{subfigure}[t]{0.18\textwidth}
			\centering
			\includegraphics[width=\textwidth]{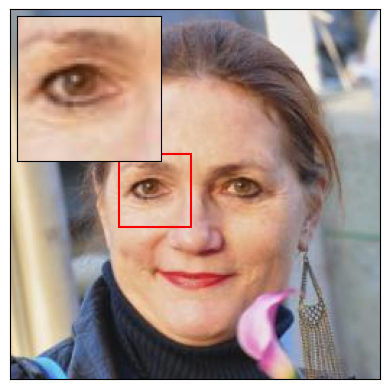}
			\subcaption{Ground truth}\label{fig:supres_pgd}
		\end{subfigure}
		\hspace{0.02cm}
		\begin{subfigure}[t]{0.18\textwidth}
			\centering
			\includegraphics[width=\textwidth]{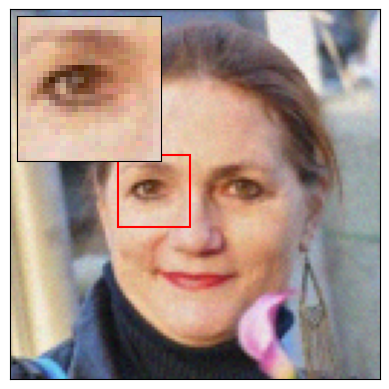}
			\subcaption{Observed}
		\end{subfigure}
		
		\vspace{0.5cm}
		
		\begin{subfigure}[t]{0.18\textwidth}
			\centering
			\includegraphics[width=\textwidth]{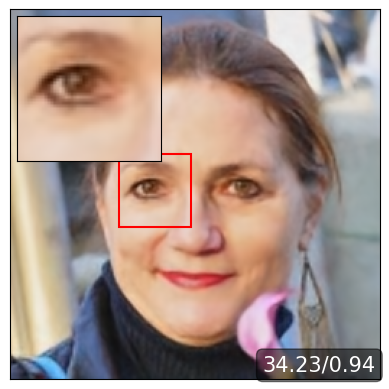}
			\subcaption{$x$ at $i_{\min}$ ($x_{\min}$)}
		\end{subfigure}
		\hspace{0.02cm}
		\begin{subfigure}[t]{0.18\textwidth}
			\centering
			\includegraphics[width=\textwidth]{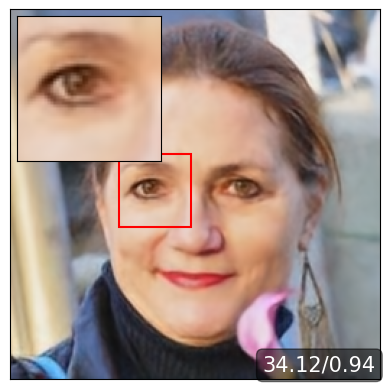}
			\subcaption{$x$ at $i_{\min}+10$}
		\end{subfigure}
		\hspace{0.02cm}
		\begin{subfigure}[t]{0.18\textwidth}
			\centering
			\includegraphics[width=\textwidth]{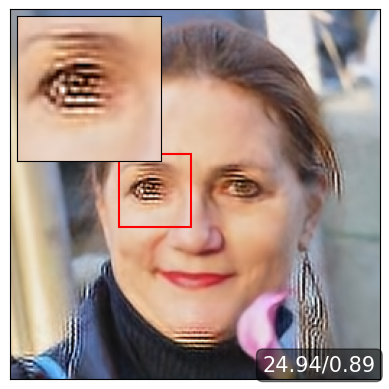}
			\subcaption{$x$ at $i_{\min}+50$}
		\end{subfigure}
		\hspace{0.02cm}
		\begin{subfigure}[t]{0.18\textwidth}
			\centering
			\includegraphics[width=\textwidth]{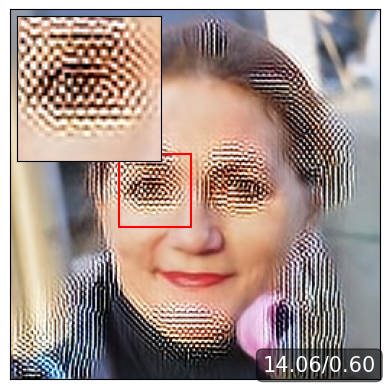}
			\subcaption{$x$ at $i_{\min}+100$}
		\end{subfigure}
		\caption{Super-resolution of images using GPGD without regularization. (a) represents recovery error for 50 images. The blue curve is associated with the recovery of (b). The quantity $x-\hat{x}$ quickly diverges after reaching the optimal solution and the images become unusable.}\label{fig:supres_pgd_graph_n_visu} \vspace{-3mm}
	\end{figure}
	
	\begin{figure}[h!]
		\centering	
		\begin{subfigure}[t]{0.36\textwidth}
			\centering
			\includegraphics[width=\textwidth]{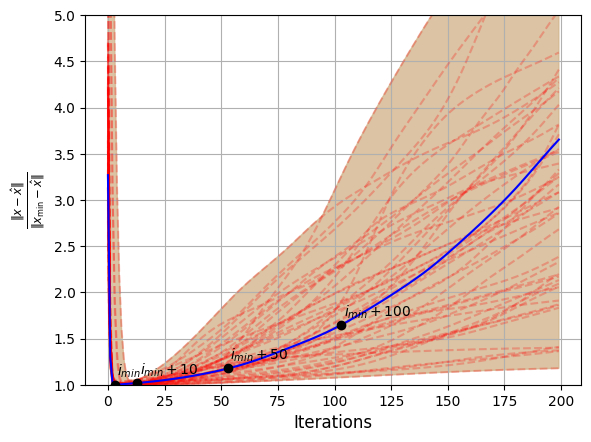}
			\subcaption{Evolution of $\frac{\|x_i-\hat{x}\|}{\|x_{\min}-\hat{x}\|}$}
		\end{subfigure}
		\hspace{0.5cm}
		\begin{subfigure}[t]{0.18\textwidth}
			\centering
			\includegraphics[width=\textwidth]{Experiments/SUPRES/39.png}
			\subcaption{Ground truth}\label{fig:supres_nipr}
		\end{subfigure}
		\hspace{0.02cm}
		\begin{subfigure}[t]{0.18\textwidth}
			\centering
			\includegraphics[width=\textwidth]{Experiments/SUPRES/39_Noisy_2_0.7_0.02.png}
			\subcaption{Observed}
		\end{subfigure}
		
		\vspace{0.5cm}
		
		\begin{subfigure}[t]{0.18\textwidth}
			\centering
			\includegraphics[width=\textwidth]{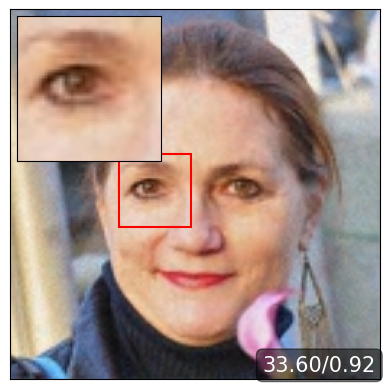}
			\subcaption{$x$ at $i_{\min}$ ($x_{\min}$)}
		\end{subfigure}
		\hspace{0.02cm}
		\begin{subfigure}[t]{0.18\textwidth}
			\centering
			\includegraphics[width=\textwidth]{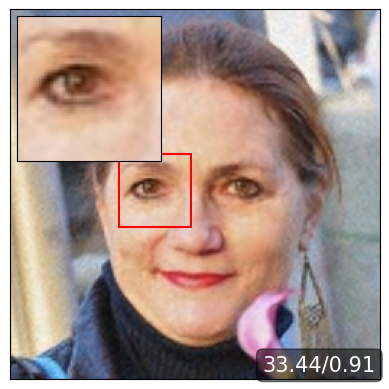}
			\subcaption{$x$ at $i_{\min}+10$}
		\end{subfigure}
		\hspace{0.02cm}
		\begin{subfigure}[t]{0.18\textwidth}
			\centering
			\includegraphics[width=\textwidth]{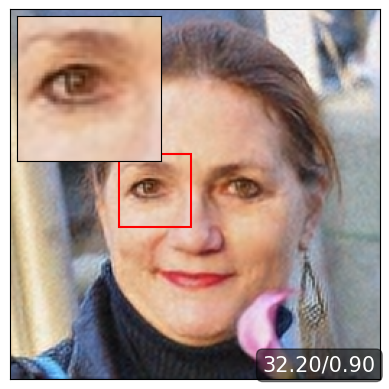}
			\subcaption{$x$ at $i_{\min}+50$}
		\end{subfigure}
		\hspace{0.02cm}
		\begin{subfigure}[t]{0.18\textwidth}
			\centering
			\includegraphics[width=\textwidth]{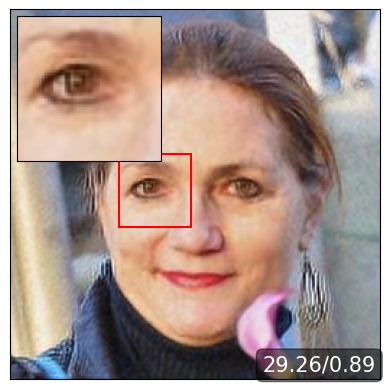}
			\subcaption{$x$ at $i_{\min}+100$}
		\end{subfigure}
		
		\caption{Super-resolution of images using NIPR.  (a) represents recovery error for 50 images. The blue curve is associated with the recovery of (b). The quantity $x-\hat{x}$ is slowly diverging after reaching the optimal solution. Yet the obtained images are still recovering the original image correctly. }\label{fig:supres_nipr_graph_n_visu}\vspace{-3mm}
	\end{figure}

	\section{Conclusion}
	We presented an extended convergence analysis of the generalized projected gradient descent algorithm by taking into account generalized backprojections, model error and projection error. The result exposes in particular how each of these factors are affected by the restricted Lipschitz constant. We showed how we can adapt the backprojection to adapt GPGD to structured noise and what trade-offs result from this adaptation. To control the projection error, we proposed a Normalized IdemPotent Regularization (NIPR) on Deep projective priors improving experimentally the stability achieved by such GPGD methods. 
	
	The results presented in this article lead to the following future works. First, automatically estimating the back-projection for different structured noises and understanding how theoretical guarantees are affected would be a natural extension of adaptation to sparse noise. Second, understanding if we could use jointly idempotent regularization, and other regularizations such as stochastic orthogonal regularization, would also be interesting.
	
	\bibliographystyle{abbrv}
	\bibliography{stable_recovery}
	
	\appendix

	\section{Proof of  Theorem~\ref{th:stability}}
	
	\begin{proof}[Proof of Theorem~\ref{th:stability}]
		For any $n$, we bound the quantity
		\begin{equation}
			\begin{split}
				\|x_{n+1} -P_\Sigma(\gt)\|_2&= \| P(x_n)- \mu L(AP(x_n)-y) -P_\Sigma(\gt) \|_2\\
				&= \| P(x_n) - \mu L A(P(x_n)-\gt) +\mu L e -P_\Sigma(\gt)\|_2\\
				&= \| P(x_n) - \mu L A(P(x_n)-P_\Sigma(\gt) +P_\Sigma(\gt)-\gt) +\mu L e -P_\Sigma(\gt)\|_2\\
				&=\| (I- \mu L A)(P(x_n)-P_\Sigma(\gt))+ \mu L A(\gt - P_\Sigma(\gt)) +\mu L e  \|_2\\
				&=\| (I- \mu L A)(P(x_n)- P_\Sigma(x_n))\\
				&+(I- \mu L A)(P_\Sigma(x_n)-P_\Sigma(\gt))+ \mu L A(\gt - P_\Sigma(\gt)) +\mu L e  \|_2\\
			\end{split}
		\end{equation}
		With the triangle inequality and the RIC, we have
		\begin{equation}
			\begin{split}
				\|x_{n+1} -P_\Sigma(\gt)\|_2 &\leq \| (I- \mu LA)(P_\Sigma(x_n)-P_\Sigma(\gt))\|_2 \\
				&+\|(I- \mu LA)(P(x_n)- P_\Sigma(x_n)) + \mu LA(\gt - P_\Sigma(\gt)) +\mu Le \|_2\\
				&\leq \delta\| P_\Sigma(x_n)-P_\Sigma(\gt)\|_2+\|(I- \mu LA)(P(x_n)- P_\Sigma(x_n)) \|_2\\
				&+\| \mu LA(P_\Sigma(\gt) -\gt)\|_2 +\| \mu Le  \|_2\\
			\end{split}
		\end{equation}
		
		With the restricted $\beta$-Lipschitz condition of $P_\Sigma$ (which implies $P_\Sigma(P_\Sigma(x)) = P_\Sigma(x)$, see \cite{traonmilin2024towards}), we have
		\begin{equation}
			\begin{split}
				\|x_{n+1} -P_\Sigma(\gt)\|_2 &\leq \delta\beta\|x_n-P_\Sigma(\gt)\| +\| I-\mu LA\|_{\mathrm{op}}\|P(x_n)- P_\Sigma(x_n) \|_2\\
				&+ \| \mu L A\|_{\mathrm{op}}\|P_\Sigma(\gt) -\gt\|_2 +\| \mu L e  \|_2\\
			\end{split}
		\end{equation}
		As $P(x_n)- P_\Sigma(x_n) = R(x_n)$, by hypothesis on $R$ and by definition of $d_{P_\Sigma}(\gt,\Sigma)$, we obtain
		\begin{equation}\label{eq:rec3b}
			\begin{split}
				\|x_{n+1} -P_\Sigma(\gt)\|_2 &\leq \delta\beta\|x_n-P_\Sigma(\gt)\| +\| I-\mu LA \|_{\mathrm{op}}\eta+ \| \mu L A\|_{\mathrm{op}}d_{P_\Sigma}(\gt,\Sigma) +\| \mu L e  \|_2\\
			\end{split}
		\end{equation}
		Let $\xi = \| I-\mu LA \|_{\mathrm{op}}\eta+ \| \mu L A\|_{\mathrm{op}}d_{P_\Sigma}(\gt,\Sigma) +\| \mu L e  \|_2 $. We show by induction
		\begin{equation}
			\begin{split}
				\|x_{n+1} -P_\Sigma(\gt)\|_2&\leq (\delta \beta)^{n+1}\|x_0-P_\Sigma(\gt)\|_2+\left(\sum_{i=0}^{n}(\delta \beta)^i\right) \xi \\
			\end{split}
		\end{equation}
		
		For $n=0$, this is exactly~\eqref{eq:rec3b}.

		Suppose step $n$ true, with \eqref{eq:rec3b} (at step $n+1$), we have:
		\begin{equation}
			\begin{split}
				\|x_{n+2} -P_\Sigma(\gt)\|_2
				& \leq \delta \beta \left((\delta \beta)^{n+1}\|x_0-P_\Sigma(\gt)\|_2+\left(\sum_{i=0}^{n}(\delta \beta)^i\right)\xi \right) +\xi\\
				& = (\delta \beta)^{n+2}\|x_0-P_\Sigma(\gt)\|_2+\left(\sum_{i=1}^{n+1}(\delta \beta)^{i}\right) \xi + \xi \\
				&= (\delta \beta)^{n+2}\|x_0-P_\Sigma(\gt)\|_2+\left(\sum_{i=0}^{n+1}(\delta \beta)^{i}\right)\xi.\\
			\end{split}
		\end{equation}
		
		This shows the induction.
		
		We also have for any $n$, when $\delta\beta <1$,
		
		\begin{equation}
			\begin{split}
				\|x_{n} -\gt\|_2&\leq \|x_{n} -P_\Sigma(\gt) + P_\Sigma(\gt) -\gt\|_2\\
				&\leq \|x_{n} -P_\Sigma(\gt)\|_2 + \| P_\Sigma(\gt) -\gt\|_2 \\
				& \leq  (\delta \beta)^{n}\|x_0-P_\Sigma(\gt)\|_2+\left(\sum_{i=0}^{n-1}(\delta \beta)^{i}\right)\xi +d(\gt,\Sigma)\\
				& \leq (\delta \beta)^{n}\|x_0-P_\Sigma(\gt)\|_2 + \frac{\| (I-\mu L A)\|_{\mathrm{op}}\eta}{1-\delta \beta}+ \left(1+\frac{\|\mu L A)\|_{\mathrm{op}}}{1-\delta \beta}\right)d(\gt,\Sigma)\\ 
				&+\frac{\|\mu L e  \|_2}{1-\delta \beta}
			\end{split}
		\end{equation}
		which concludes the proof.
	\end{proof}

	\section{A numerical illustration of the importance of restricted Lipschitz for stable recovery}\label{sec:exp_lip_approx}
	
	We consider the problem of sparse recovery, i.e. $\Sigma = \Sigma_k$ the set of $k$-sparse vectors with PGD. It has been shown that the orthogonal projection (which amounts to performing iterative hard thresholding) is restricted Lipschitz with constant $\beta= \sqrt{\frac{3 + \sqrt{5}}{2}}$ which is close to optimal for restricted Lipschitz projections onto $\Sigma_k$ \cite{traonmilin2024towards}. We propose to consider projections $P_\alpha$ that deteriorate the restricted Lipschitz constant of the orthogonal projection. We define, for any $z \in \bR^N$:
	\begin{equation} \label{eq:def_P_alpha}
		P_\alpha(z):= 
		\left\{
		\begin{array}{l}
			\left(1 + \alpha \frac{\|z-P_\Sigma^\perp(z)\|_2}{\|P_\Sigma^\perp(z)\|_2}\right) P_\Sigma^\perp(z) \ \mbox{ if $P_\Sigma^\perp(z) \neq 0$.}
			\\
			0 \ \mbox{ otherwise.}
		\end{array}
		\right.
	\end{equation}
	where the orthogonal projection is defined in Definition~\ref{def:orth_proj} (and is the hard thresholding operator when the model is $\Sigma_k$). 
	
	\begin{lemma} \label{lem:approx_orth_proj}
		Let $\Sigma = \Sigma_k$. Consider $P_\alpha$ defined in~\eqref{eq:def_P_alpha} then 
		\begin{equation}
			\beta_{\Sigma}(P_\alpha) \leq \beta_{\Sigma}(P_{\Sigma}^\perp) + \alpha.
		\end{equation}
	\end{lemma}
	
	\begin{proof}[Proof of Lemma~\ref{lem:approx_orth_proj}]
		For any $z$, $x\in \Sigma$, we have, by definition of the restricted $\beta$-Lipschitz property (of $P_\Sigma^\perp$)
		
		\begin{equation}
			\begin{split}
				\|P_\alpha(z) -x\|_2 &\leq \|P_\alpha(z) -P_\Sigma^\perp(z)\|_2 +\|P_\Sigma^\perp(z) -x\|_2 \\
				&\leq \alpha \| \frac{\|z-P_\Sigma^\perp(z)\|_2}{\|P_\Sigma^\perp(z)\|_2} P_\Sigma^\perp(z)\| + \beta_\Sigma(P_\Sigma^\perp)\|z -x\|_2 \\
				&= \alpha \|z-P_\Sigma^\perp(z)\|_2+ \beta_\Sigma(P_\Sigma^\perp)\|z -x\|_2\\
			\end{split}
		\end{equation}
		By definition of $P_\Sigma^\perp(z)$, as $x\in\Sigma$, we have $\|z-P_\Sigma^\perp(z)\|_2 \leq \|z-x\|_2$ and 
		
		\begin{equation}
			\begin{split}
				\|P_\alpha(z) -x\|_2 &= (\alpha + \beta_\Sigma(P_\Sigma^\perp))\|z -x\|_2.\\
			\end{split}
		\end{equation}
	\end{proof}
	
	We remark that $\beta_\Sigma(P_\alpha) \to_{\alpha\to 0} \beta_\Sigma(P_\Sigma^\perp)$, and that we control the Lipschitz constant $P_\alpha$ with the parameter $\alpha$.
	
	In Figure~\ref{fig:app_iht} (top), we perform stable sparse recovery experiments with different sparsities, with IHT (PGD with $P_0 = P_\Sigma^\perp$) and PGD with $P_\alpha$ ($\alpha\neq 0$). For each considered sparsity of $\gt$, we perform $50$ experiments and plot the normalized $\ell^2$ reconstruction error (thresholded by $1$) of the $95\%$ centile. We observe that increasing $\alpha$ and thus degrading $\beta$ diminishes the identifiability properties of PGD with $P_\alpha$.
	
	\begin{figure}[h!]
		\centering
		\includegraphics[width=.75\linewidth]{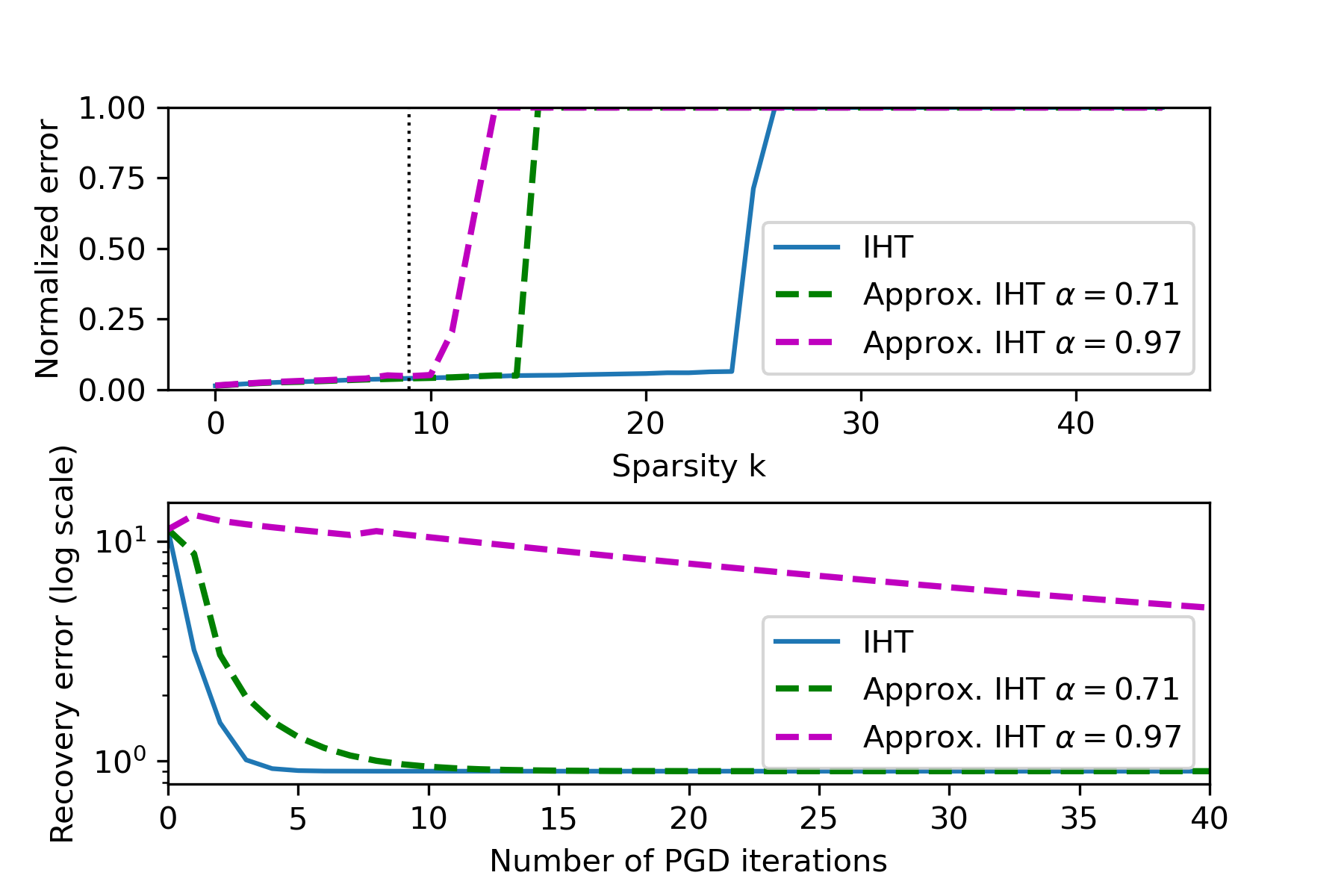}
		\caption{Importance of the restricted Lipschitz constant for stable recovery: the case of sparse recovery. Top: Normalized error bound for 95$\%$ of the experiments with respect to the sparsity $k$ of the unknown. Bottom: convergence for one experiment with $k=9$. We observe that worsening the Lipschitz constant (through the parameter $\alpha$) deteriorates both the convergence rate and the identifiability properties of the algorithm.}
		
		\label{fig:app_iht}
	\end{figure}
	
	For a given sparsity, where we observe stable recovery for all $\alpha$ (Figure~\ref{fig:app_iht} (bottom)), we show the convergence of the different PGD algorithms for a fixed step $\mu$ (chosen to be the largest to obtain convergence of IHT). We observe that the convergence rate is decreased when $\alpha$ increases, thus matching Theorem~\ref{th:stability}. Note that the stability is not changed in these experiments. We attribute this to the fact that the support of $\gt$ is necessarily identified for stable recovery, leading to $P_\alpha(z) \approx P_\Sigma^\perp(z)$ when $z$ is close to $\Sigma$.

	\section{Impact of the step size $\mu$ on the stability constant} \label{sec:impact_mu}
	Let us consider the case $L = A^T$. In this case, we have $C_{\mathrm{stab}} := \frac{\mu}{1-\delta(\mu A^TA) \beta} $. As $\delta(\mu A^TA)$ can be interpreted as an operator norm of $I-\mu A^T A $ \emph{restricted} to the low-dimensional model $\Sigma$, it is not clear how $\delta(\mu A^TA)$ behaves with respect to $\mu$ (except that $\lim_{\mu\to0}\delta(\mu A^TA) =1$ and that we do not verify convergence hypotheses for small $\mu$).  In the context of sparse recovery where PGD with the orthogonal projection is  Iterative Hard Tresholding, we illustrate a trade-off between convergence speed and quality of recovery. Indeed, while we generally look a the largest possible $\mu$ for fast convergence (when interpreted as a gradient step), we observe that lowering $\mu$ can improve stability at the expense of identifiability and convergence speed. In Figure~\ref{fig:influence_mu}, we represent the recovery error with respect to sparsity of the worst 10th centile for sparse recovery ($m=150,n=300$, noisy random Gaussian measurements with fixed noise variance). We plot the convergence of recovery error $\|x^*-\hat{x}\|_2$ for a fixed sparsity. We observe that the best $\mu$ for stable recovery $\mu =0.6$ allows for stable recovery of sparsities $\leq 15$ while the best $\mu = 0.3$ improves noise stability for sparsity $k=4$ at the expense of reduced identifiability and convergence speed.

	\begin{figure}[h]
		\centering
		\includegraphics[width=.85\linewidth]{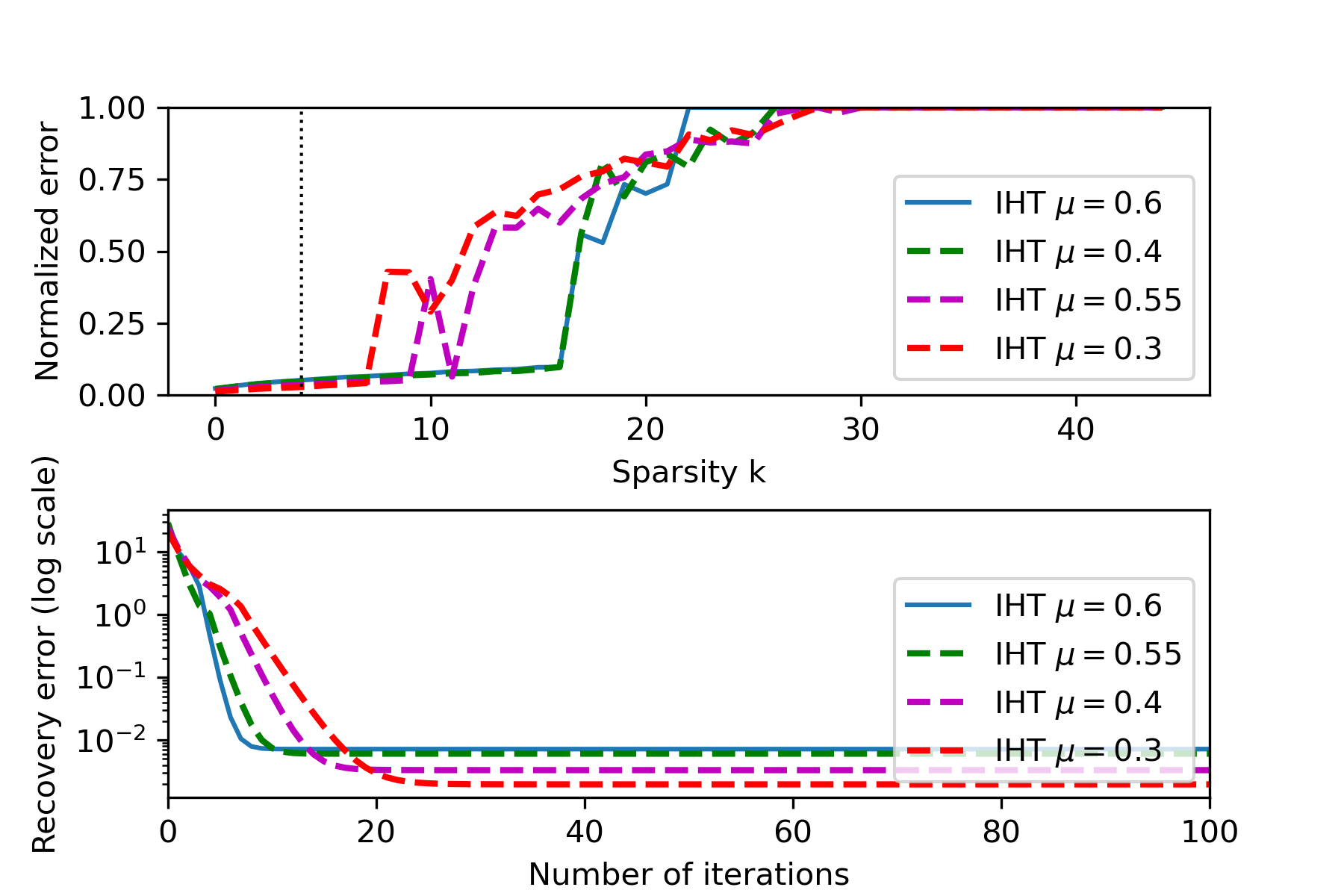}
		\caption{Impact of the step size on the stability of PGD for sparse recovery. Top: phase transition diagram for stable recovery. Bottom: impact of $\mu$ on convergence for $k=4$}
		
		\label{fig:influence_mu}
	\end{figure}
	
	\section{Trade-offs for stable linear sparse recovery with GPGD (IHT)}\label{sec:tradeoff_sparse}
	
	In the context of sparse recovery and random Gaussian measurements, a restricted isometry property of $A^TA$ is guaranteed with high probability under the condition that the number of measurements $m \geq C k \log(n/k)$ for some potentially large constant $C$. As $S$ selects $m-s$ measurements in $A$, we have that $SA$ has a restricted isometry with high probability if $m \geq s+ Ck \log(m/k)$. Qualitatively, fast stable recovery with iterative hard thresholding is possible if the number of measurements is $O(s+k)$ and the trade-off between identifiability of sparse vectors and robustness to sparse noise is explicit. We discuss in Section~\ref{sec:joint_proj} how to adapt GPGD  to unknown noise support by recasting the problem with joint models of signal and noise. 
	
	In Figure~\ref{fig:stable_iht}, we illustrate the trade-off between noise adaptation and identifiability in the case of sparse recovery with a random measurement operator. For different values of sparsity $k$, we add a high amplitude outlier noise of sparsity $s$ to a low energy Gaussian noise and perform a $m-s$ hard thresholding of the residual $Ax_n-y$ in PGD (i.e. iterative hard thresholding). As predicted by the theory, the trade-off between $k$ and $s$ drives the success of the algorithm through the restricted isometry constant of $\mu A^TSA$ (here $S$ is the support selected by the hard thresholding operator on the residual).
	
	\begin{figure}[h!]
		\centering
		\includegraphics[width=.75\linewidth]{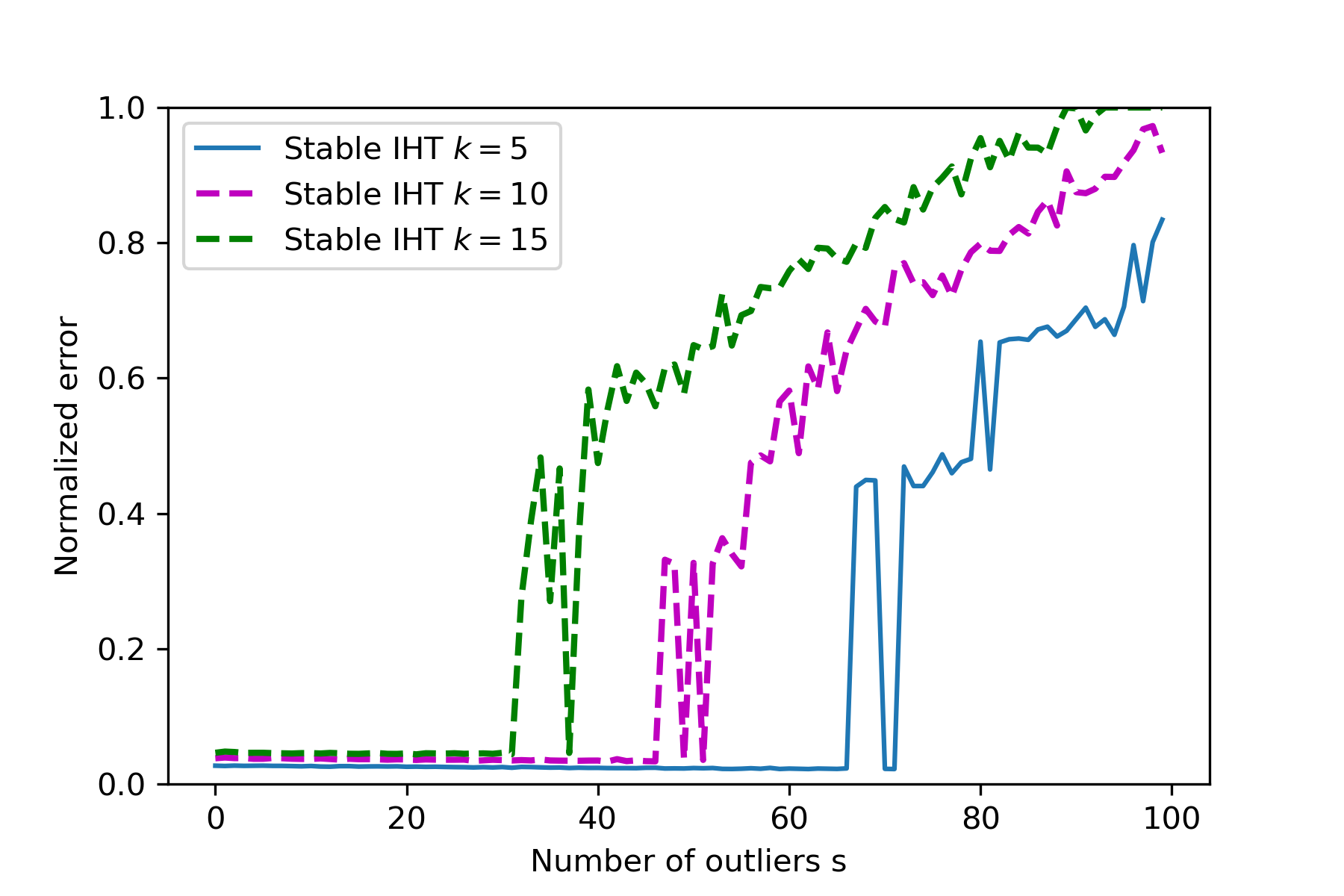}
		\caption{Adaptation to sparse noise in IHT and sparse recovery: Illustration of the trade-off between stability to noise sparsity and identifiability of sparse vectors. We show the Normalized error bound for 90$\%$ of the experiments with respect to the sparsity $s$ of outliers for three different sparsities. The sparser the noise is, the greater the identifiability properties of stable IHT. }
		
		\label{fig:stable_iht}
	\end{figure}

	\section{A remark on sparse noise estimation with GPGD} \label{sec:joint_proj}
	When the support of the sparse noise is not known, note that the problem of robust low-dimensional recovery can be recast as a simple low-dimensional recovery problem where we consider $\tilde{\Sigma} = \Sigma \times \Sigma_{noise}$ and
	\begin{equation}
		\tilde{y} = \tilde{A}\tilde{x}
	\end{equation}
	where $\tilde{x} = (\gt^T,e^T)^T$, $\tilde{A} = (A,I)$.
	The algorithm then estimates at each step both the unknown and the noise. The RIP condition notoriously relies on an incoherence between noise and sparsity model (i.e. separated support of noise and gradients for gradient-sparse image recovery~\cite{traonmilin2015robust}). It is natural to study the question of the choice of optimal $P_{\tilde{\Sigma}}$ in this case, as the low-dimensional model for signal and noise can be defined as a product model. In this case, we can construct the optimal projection for the restricted Lipschitz constant by simply concatenating optimal projections for each model set. Recall that we can define the $\ell^2$-norm in a product of Euclidean spaces $E_1 \times E_2 $ by $\|(x_1,x_2)\|_2 := \sqrt{\|x_1\|_2^2+\|x_2\|_2^2}$ for all $(x_1,x_2) \in E_1 \times E_2$. We have the following Lemma.

	\begin{lemma}\label{lem:product_model_proj}
		Let $\Sigma = \Sigma_1 \times \Sigma_2$, $P_i \in \arg \min_P \beta_{\Sigma_i}(P)$. Consider $P_\Sigma (z) = P_\Sigma (z_1,z_2)= (P_1(z_1), P_2(z_2))$. Then $P_\Sigma \in \arg \min_P \beta_{\Sigma}(P)$.
	\end{lemma}

	\begin{proof}[Proof of Lemma~\ref{lem:product_model_proj}]
		Suppose, w.l.o.g that $\beta_2 \geq \beta_1$. Let $z \in \bR^{N_1} \times \bR^{N_2}, x \in \Sigma$. We have
		\begin{equation}
			\begin{split}
				\|P_\Sigma(z)-x\|_2^2 &= \|P_1(z_1)-x_1\|_2^2 +\|P_2(z_2)-x_2\|_2^2 \leq \beta_1^2 \|z_1-x_1\|_2^2 +\beta_2^2\|z_2-x_2\|_2^2\\
				&\leq \max(\beta_1^2,\beta_2^2)(\|z_1-x_1\|_2^2 +\|z_2-x_2\|_2^2) = \beta_2^2\|z-x\|_2^2\\
			\end{split}
		\end{equation}
		We have thus shown that $\beta_\Sigma(P_\Sigma) \leq \beta_2$.
		
		Now let $Q$ be a generalized projection such that $\beta_\Sigma(Q) < \beta_2$. We have, for all $z \in E$,
		
		\begin{equation}
			\begin{split}
				\|Q(z)-x\|_2^2 &\leq \beta_\Sigma(Q)^2\|z-x\|_2^2\\
			\end{split}
		\end{equation}
		Consider $x = (x_1,x_2)\in \Sigma$, $z = (x_1,z_2) \in \Sigma_1 \times \bR^{N_2}$. We have
		
		\begin{equation}
			\begin{split}
				\|Q(z)-x\|_2^2 = \|[Q(z)]_1-x_1\|_2^2 +\|[Q(z)]_2-x_2\|_2^2 &\leq \beta_\Sigma(Q)^2(\|x_1-x_1\|_2^2 +\|z_2-x_2\|_2^2)\\
				&= \beta_\Sigma(Q)^2 \|z_2-x_2\|_2^2\\
			\end{split}
		\end{equation}
		
		Consider the application $\tilde{P}_2 : \bR^{N_2} \to \bR^{N_2}$ defined by $\tilde{P}_2(z_2) = [Q(x_1,z_2)]_2 \in \Sigma_2$. We deduce that
		
		\begin{equation}
			\begin{split}
				\|\tilde{P}_2(z_2)-x_2\|_2^2 = \|[Q(z)]_2-x_2\|_2^2 &\leq  \beta_\Sigma(Q)^2 \|z_2-x_2\|_2^2\\
			\end{split}
		\end{equation}
		We deduce that $\beta_{\Sigma_2}(\tilde{P}_2) \leq \beta_\Sigma(Q) < \beta_2 = \beta_{\Sigma_2}^\star$, where $\beta_{\Sigma_2}^\star$ is the optimal restricted Lipschitz constant (by hypothesis), which is impossible.
		
		We deduce that $\beta_\Sigma(Q) =\beta_{\Sigma}^\star$.
	\end{proof}
	We can generalize to any number of product models with the following corollary.
	
	\begin{corollary}\label{cor:cor_prod_model}
		Let $\Sigma = \Sigma_1 \times \ldots \times \Sigma_q$, $P_i \in \arg \min_P \beta_{\Sigma_i}(P)$. Consider $P_\Sigma (z) = P_\Sigma (z)= (P_1(z_1), \ldots, P_q(z_q))$. Then $P_\Sigma \in \arg \min_P \beta_{\Sigma}(P)$.
	\end{corollary}
	
	\begin{proof}[Proof of Corollary~\ref{cor:cor_prod_model}]
		By induction on $q$, for $q= 2$, use Lemma~\ref{lem:product_model_proj}.
		
		Suppose this corollary true for some $q$, for $q+1$ consider $\tilde{\Sigma}_1 = \Sigma_1 \times \ldots \times \Sigma_q$ and $\tilde{\Sigma_{2}} = \Sigma_{q+1}$. Apply the corollary to  $\tilde{\Sigma}_1$ and Lemma~\ref{lem:product_model_proj} to $\tilde{\Sigma}_1 \times \tilde{\Sigma_{2}}$.
	\end{proof}
	
	Note that we took the example of sparse corruptions of sparse models. In some applications, such as low-rank models and sparse noise (or sparse models and low-rank noise), restricted isometries can be obtained if there is sufficient incoherence between the sparsity model and the low-rank model. In a learning context, jointly learning projective priors for additive models has been explored in the context of structure-texture decomposition~\cite{guennec2025joint}.

	\section{Additional experiments} \label{sec:additional_exp}
	\subsection{Autoencoders}
	We propose in this subsection to train autoencoders over the MNIST dataset with and without NIPR. The size of the train set is 30000. To test our regularization for these, we propose to solve an inpainting inverse problem, and we display the graphs of $\frac{\|x-\hat{x}\|}{\|\hat{x}\|}$. Figure \ref{fig:MNIST} shows that NIPR led to a more stable convergence as the curves are not increasing as much as for GPGD after reaching the optimal solution.

	\begin{figure}[!h]
		\centering
		\begin{subfigure}{0.40\textwidth}
			\includegraphics[width=1\textwidth]{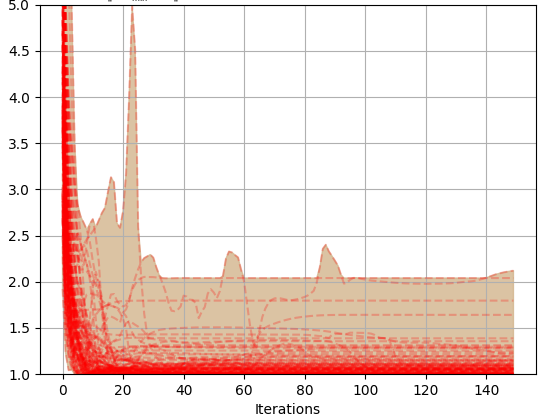}
			\caption{No reg.}
		\end{subfigure}
		\hspace{1cm}
		\begin{subfigure}{0.40\textwidth}
			\includegraphics[width=1\textwidth]{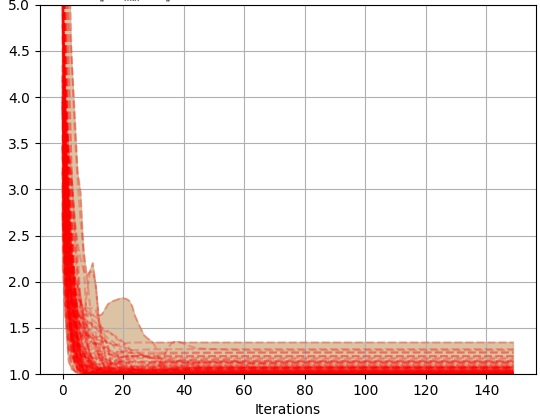}
			\caption{NIPR}
		\end{subfigure}

		\begin{subfigure}{0.8\textwidth}
			\includegraphics[width=1\textwidth]{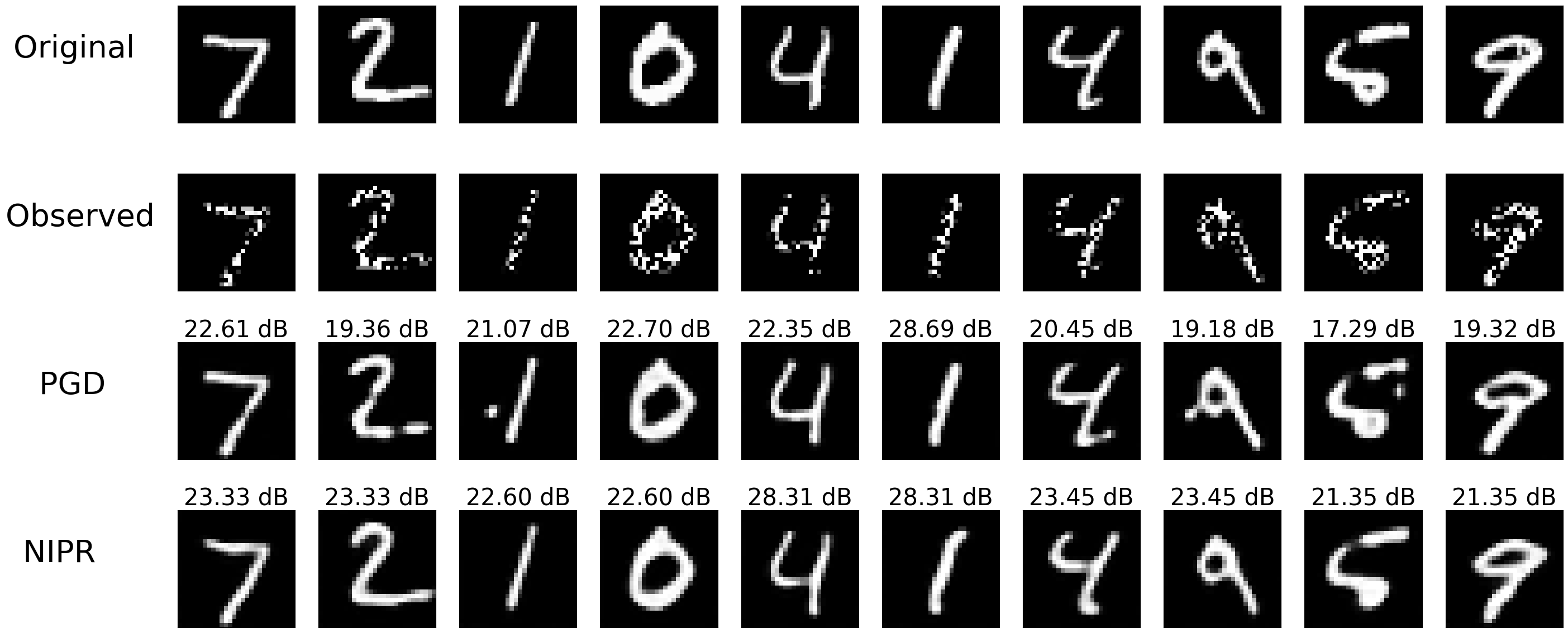}
			\caption{}
		\end{subfigure}

		\caption{Inpainting of MNIST images using autoencoders with and without NIPR regularization. Once again, and for another DPP, the NIPR led to a more stable convergence compared to when we only use the vanilla GPGD algorithm.}\label{fig:MNIST}
	\end{figure}
	
	This shows in particular that, similarly to denoisers, NIPR can act on the stability of GPGD for another DPP, an autoencoder here.

\newpage 	
	\subsection{Additional experiments on denoisers}
	We propose in this subsection to consider additional inverse problems experiments using denoisers in their GPGD algorithms. Therefore, we consider the deblurring and inpainting tasks.
	
	\subsubsection{Deblurring}
	\begin{table}[h!]	
		\caption{PSNRs and stability values for a deblurring inverse problem using GPGD without regularization, NIPR and SOR. While recovering the original image correctly, NIPR is clearly more stable after reaching $x_{\min}$ compared to GPGD and SOR.}\label{tab:denoising_fig}
		\begin{minipage}{0.48\textwidth}
			\centering
			
			\begin{tabular}{ccccccccc}
				\toprule
				\multirow{2}{*}{Method} & \multirow{2}{*}{PSNR$\uparrow$} & \multirow{2}{*}{SSIM$\uparrow$} & \multicolumn{3}{c}{SM1 $\downarrow$}  & \multicolumn{3}{c}{SM2 $\downarrow$}  \\
				
				&            &            & $i_{\min}+10$ & $i_{\min}+50$ & $i_{\min}+100$ & $i_{\min}+10$ & $i_{\min}+50$ & $i_{\min}+100$ \\
				\hline
				No reg.           & \textbf{29,141}        & \textbf{0,832}         & 0,0014  & 1,5981            & 8,6641   & \textbf{0,060}   & 0,278   & 1,259   \\
				NIPR          & 28,946        & 0,819         & \textbf{0,0008}  & \textbf{0,1795}            & \textbf{0,8984}   & 0,068   &\textbf{ 0,133}   &\textbf{ 0,234}   \\
				SOR           & 28,147        & 0,797         & 0,0820  & 16,0790            & 37,2008  & 0,064   & 1,331   & 2,745 \\
				\hline   
			\end{tabular}
		\end{minipage}
	\end{table}

	\begin{figure}[h!]
		\centering
		
		\begin{subfigure}[t]{0.36\textwidth}
			\centering
			\includegraphics[width=\textwidth]{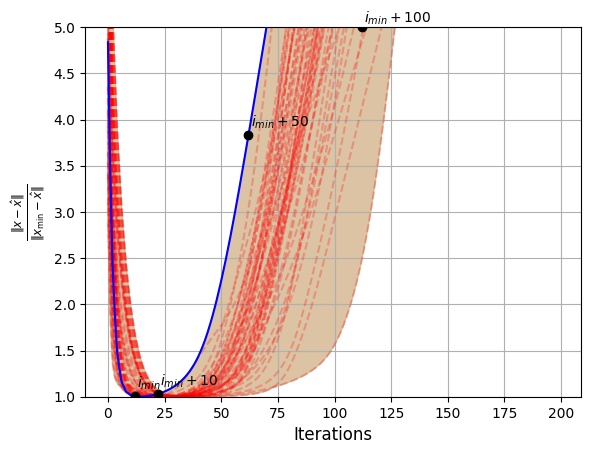}
			\subcaption{Evolution of $\frac{\|x-i-\hat{x}\|_2}{\|x_{\min}-\hat{x}\|_2}$  for 50 images. 
				The blue curve is associated to the recovery of image (b)}
		\end{subfigure}
		\hspace{0.5cm}
		\begin{subfigure}[t]{0.18\textwidth}
			\centering
			\includegraphics[width=\textwidth]{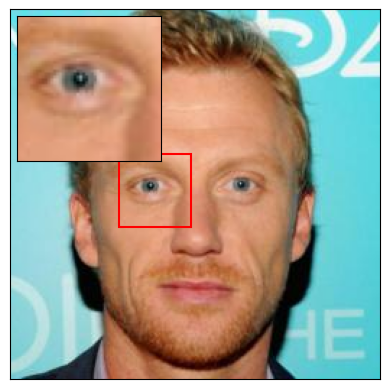}
			\subcaption{Ground truth}\label{fig:supres_pgd}
		\end{subfigure}
		\hspace{0.02cm}
		\begin{subfigure}[t]{0.18\textwidth}
			\centering
			\includegraphics[width=\textwidth]{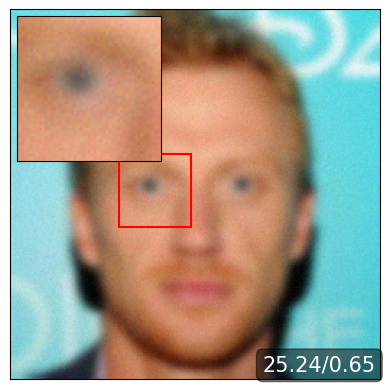}
			\subcaption{Observed}
		\end{subfigure}
		
		\vspace{0.5cm}
		
		\begin{subfigure}[t]{0.18\textwidth}
			\centering
			\includegraphics[width=\textwidth]{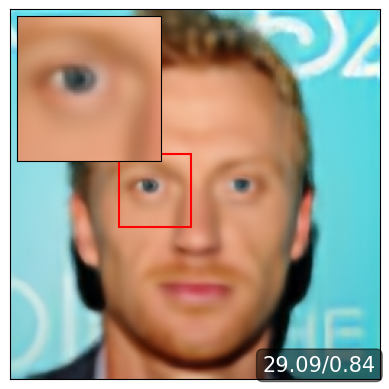}
			\subcaption{$x$ at $i_{\min}$ ($x_{\min}$)}
		\end{subfigure}
		\hspace{0.02cm}
		\begin{subfigure}[t]{0.18\textwidth}
			\centering
			\includegraphics[width=\textwidth]{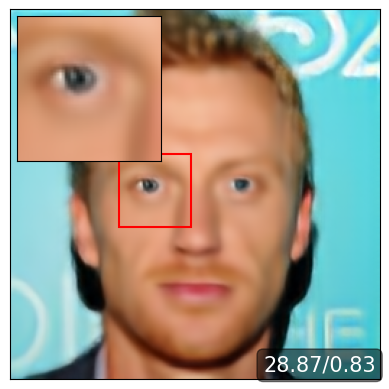}
			\subcaption{$x$ at $i_{\min}+10$}
		\end{subfigure}
		\hspace{0.02cm}
		\begin{subfigure}[t]{0.18\textwidth}
			\centering
			\includegraphics[width=\textwidth]{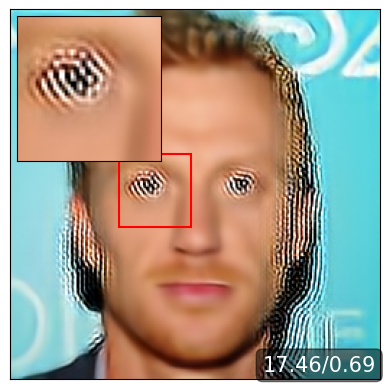}
			\subcaption{$x$ at $i_{\min}+50$}
		\end{subfigure}
		\hspace{0.02cm}
		\begin{subfigure}[t]{0.18\textwidth}
			\centering
			\includegraphics[width=\textwidth]{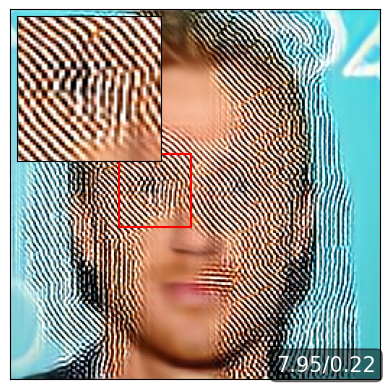}
			\subcaption{$x$ at $i_{\min}+100$}
		\end{subfigure}
		
		\caption{Deblurring of images using a GPGD algorithm without regularization. The quantity $x-\hat{x}$ quickly diverges after reaching the optimal solution and the images become unusable with several artifacts.}\label{fig:Gblur_pgd_graph_n_visu}
	\end{figure}
	
	\begin{figure}[h!]
		\centering	
		\begin{subfigure}[t]{0.36\textwidth}
			\centering
			\includegraphics[width=\textwidth]{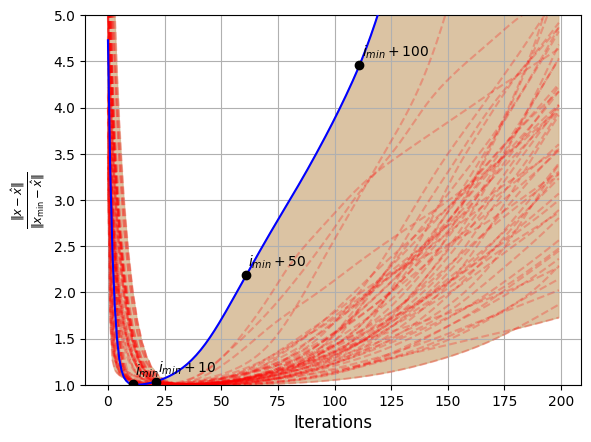}
			\subcaption{Evolution of $\frac{\|x_i-\hat{x}\|_2}{\|x_{\min}-\hat{x}\|_2}$  for 50 images. 
				The blue curve is associated to the recovery of image (b)}
		\end{subfigure}
		\hspace{0.5cm}
		\begin{subfigure}[t]{0.18\textwidth}
			\centering
			\includegraphics[width=\textwidth]{Experiments/Gblur/48.png}
			\subcaption{Ground truth}\label{fig:supres_nipr}
		\end{subfigure}
		\hspace{0.02cm}
		\begin{subfigure}[t]{0.18\textwidth}
			\centering
			\includegraphics[width=\textwidth]{Experiments/Gblur/48_Noisy_0_0.02.png}
			\subcaption{Observed}
		\end{subfigure}
		
		\vspace{0.5cm}
		
		\begin{subfigure}[t]{0.18\textwidth}
			\centering
			\includegraphics[width=\textwidth]{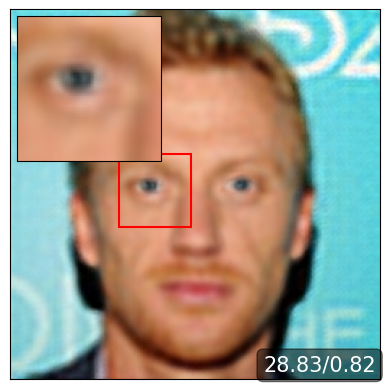}
			\subcaption{$x$ at $i_{\min}$ ($x_{\min}$)}
		\end{subfigure}
		\hspace{0.02cm}
		\begin{subfigure}[t]{0.18\textwidth}
			\centering
			\includegraphics[width=\textwidth]{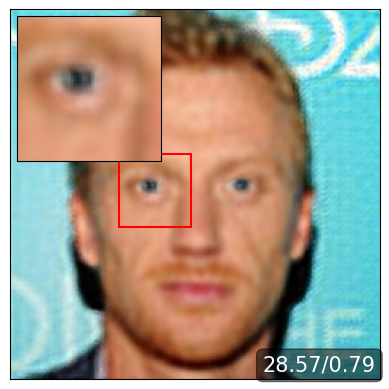}
			\subcaption{$x$ at $i_{\min}+10$}
		\end{subfigure}
		\hspace{0.02cm}
		\begin{subfigure}[t]{0.18\textwidth}
			\centering
			\includegraphics[width=\textwidth]{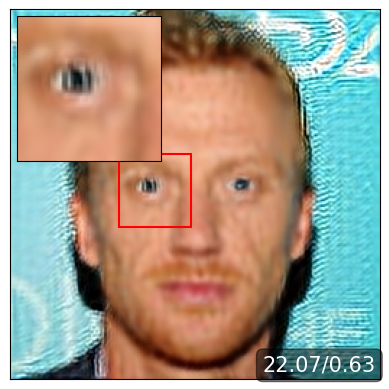}
			\subcaption{$x$ at $i_{\min}+50$}
		\end{subfigure}
		\hspace{0.02cm}
		\begin{subfigure}[t]{0.18\textwidth}
			\centering
			\includegraphics[width=\textwidth]{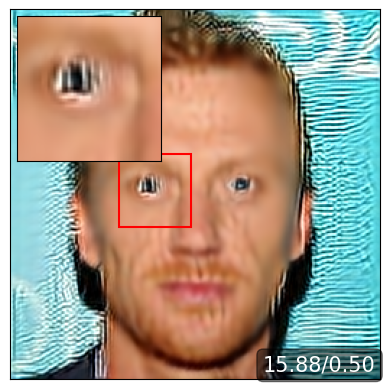}
			\subcaption{$x$ at $i_{\min}+100$}
		\end{subfigure}
		
		\caption{Deblurring of images using NIPR. The quantity $x-\hat{x}$ is diverging after reaching the optimal solution. The images degrade more slowly than GPGD without regularization. }\label{fig:Gblur_nipr_graph_n_visu}
	\end{figure}

	\newpage
	\subsubsection{Inpainting}
	
	\begin{table}[h!]	
		\caption{PSNRs and stability values for an inpainting inverse problem using GPGD without regularization, NIPR and SOR. While recovering the original image correctly, NIPR is clearly more stable after reaching $x_{\min}$ compared to GPGD without regularization. For this particular inverse problem, SOR maintains a good stability.}\label{tab:denoising_fig}
		\begin{minipage}{0.48\textwidth}
			\centering
			
			\begin{tabular}{ccccccccc}
				\toprule
				\multirow{2}{*}{Method} & \multirow{2}{*}{PSNR$\uparrow$} & \multirow{2}{*}{SSIM$\uparrow$} & \multicolumn{3}{c}{SM1 $\downarrow$}  & \multicolumn{3}{c}{SM2 $\downarrow$}  \\
				
				&            &            & $i_{\min}+10$ & $i_{\min}+50$ & $i_{\min}+100$ & $i_{\min}+10$ & $i_{\min}+50$ & $i_{\min}+100$ \\
				\hline
				No reg.           & 35,615        & 0,956         & 0,0039  & 0,0727            & 0,1376   & 0,088   & 0,468   & 1,197   \\
				NIPR          & \textbf{36,763}        & \textbf{0,959 }        & \textbf{0,0002}  & 0,0005            & 0,0017   & \textbf{0,037}   & \textbf{0,098}   & 0,134   \\
				SOR           & 34,870        & 0,945         & \textbf{0,0002}  & \textbf{0,0002}            & \textbf{0,0002}   & 0,054   & 0,099   & \textbf{0,113} \\
				\hline   
			\end{tabular}
		\end{minipage}
	\end{table}

	\begin{figure}[h!]
		\centering
		
		\begin{subfigure}[t]{0.36\textwidth}
			\centering
			\includegraphics[width=\textwidth]{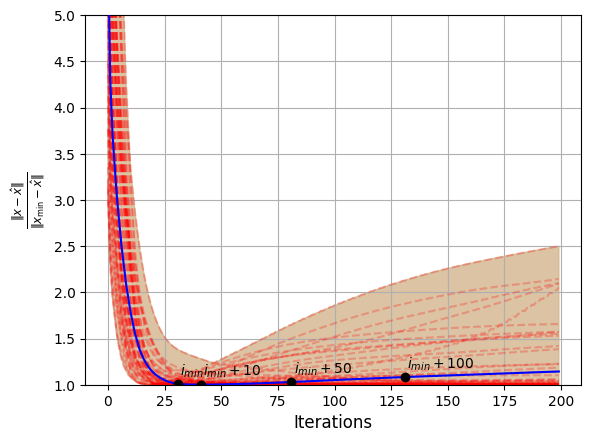}
			\subcaption{Evolution of $\frac{\|x-\hat{x}\|}{\|x_{\min}-\hat{x}\|}$for 50 images. 
				The blue curve is associated to the recovery of image (b)}
		\end{subfigure}
		\hspace{0.5cm}
		\begin{subfigure}[t]{0.18\textwidth}
			\centering
			\includegraphics[width=\textwidth]{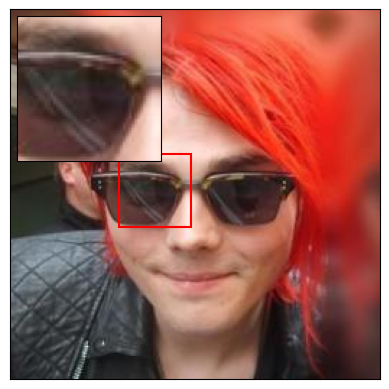}
			\subcaption{Ground truth}\label{fig:supres_pgd}
		\end{subfigure}
		\hspace{0.02cm}
		\begin{subfigure}[t]{0.18\textwidth}
			\centering
			\includegraphics[width=\textwidth]{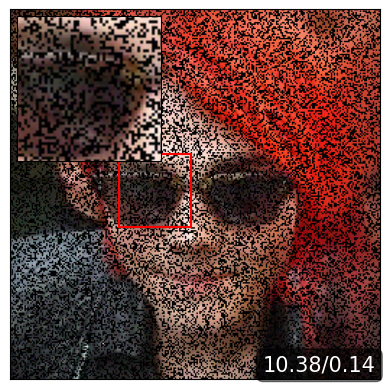}
			\subcaption{Observed}
		\end{subfigure}
		
		\vspace{0.5cm}
		
		\begin{subfigure}[t]{0.18\textwidth}
			\centering
			\includegraphics[width=\textwidth]{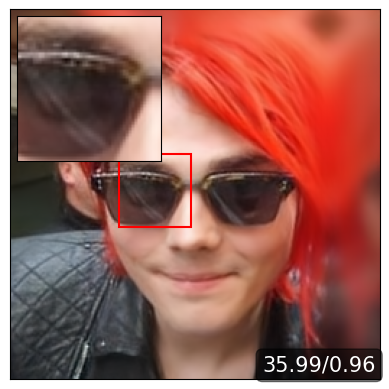}
			\subcaption{$x$ at $i_{\min}$ ($x_{\min}$)}
		\end{subfigure}
		\hspace{0.02cm}
		\begin{subfigure}[t]{0.18\textwidth}
			\centering
			\includegraphics[width=\textwidth]{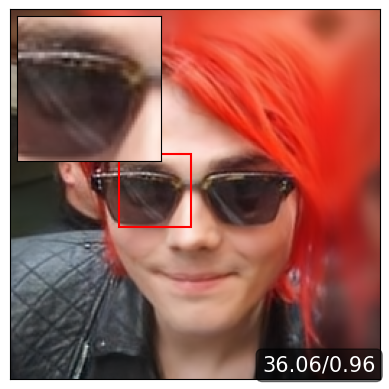}
			\subcaption{$x$ at $i_{\min}+10$}
		\end{subfigure}
		\hspace{0.02cm}
		\begin{subfigure}[t]{0.18\textwidth}
			\centering
			\includegraphics[width=\textwidth]{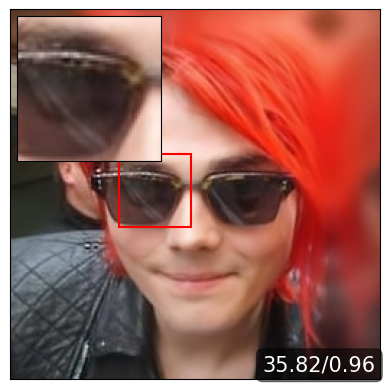}
			\subcaption{$x$ at $i_{\min}+50$}
		\end{subfigure}
		\hspace{0.02cm}
		\begin{subfigure}[t]{0.18\textwidth}
			\centering
			\includegraphics[width=\textwidth]{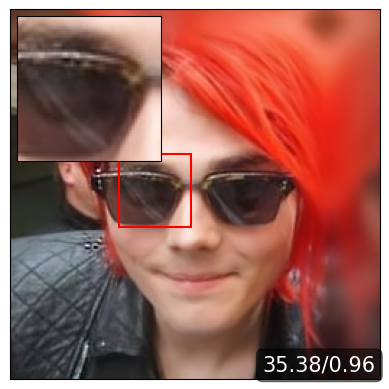}
			\subcaption{$x$ at $i_{\min}+100$}
		\end{subfigure}
		
		\caption{Inpainting of images using a PGD algorithm. The quantity $x-\hat{x}$ is slowly increasing after reaching the optimal solution but the images can still be used on average. }\label{fig:Gblur_pgd_graph_n_visu}
	\end{figure}
	
	\begin{figure}[h!]
		\centering	
		\begin{subfigure}[t]{0.36\textwidth}
			\centering
			\includegraphics[width=\textwidth]{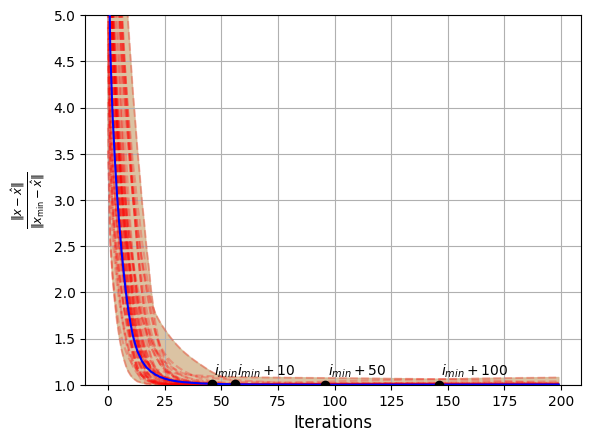}
			\subcaption{Evolution of $\frac{\|x-\hat{x}\|}{\|x_{\min}-\hat{x}\|}$ for 50 images. 
				The blue curve is associated to the recovery of image (b)}
		\end{subfigure}
		\hspace{0.5cm}
		\begin{subfigure}[t]{0.18\textwidth}
			\centering
			\includegraphics[width=\textwidth]{Experiments/Inpainting/4.png}
			\subcaption{Ground truth}\label{fig:supres_nipr}
		\end{subfigure}
		\hspace{0.02cm}
		\begin{subfigure}[t]{0.18\textwidth}
			\centering
			\includegraphics[width=\textwidth]{Experiments/Inpainting/4_Noisy_0.4_0.02.png}
			\subcaption{Observed}
		\end{subfigure}
		
		\vspace{0.5cm}
		
		\begin{subfigure}[t]{0.18\textwidth}
			\centering
			\includegraphics[width=\textwidth]{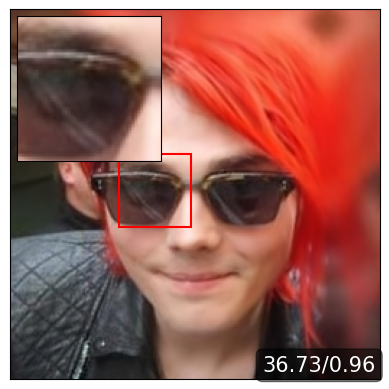}
			\subcaption{$x$ at $i_{\min}$ ($x_{\min}$)}
		\end{subfigure}
		\hspace{0.02cm}
		\begin{subfigure}[t]{0.18\textwidth}
			\centering
			\includegraphics[width=\textwidth]{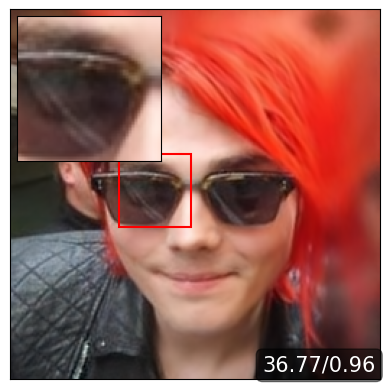}
			\subcaption{$x$ at $i_{\min}+10$}
		\end{subfigure}
		\hspace{0.02cm}
		\begin{subfigure}[t]{0.18\textwidth}
			\centering
			\includegraphics[width=\textwidth]{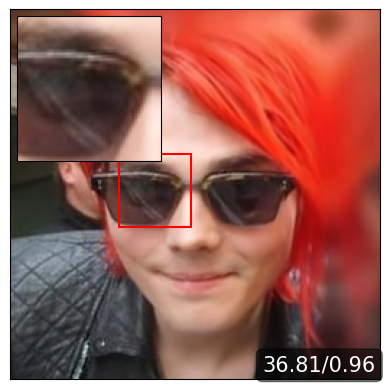}
			\subcaption{$x$ at $i_{\min}+50$}
		\end{subfigure}
		\hspace{0.02cm}
		\begin{subfigure}[t]{0.18\textwidth}
			\centering
			\includegraphics[width=\textwidth]{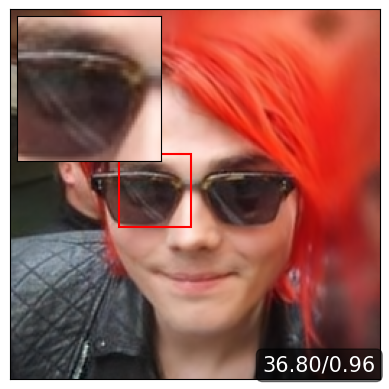}
			\subcaption{$x$ at $i_{\min}+100$}
		\end{subfigure}
		
		\caption{Inpainting of images using NIPR. The quantity $x-\hat{x}$ is converges and is stable. The images at the end of the iterations represent the original image correctly. }\label{fig:Gblur_nipr_graph_n_visu}
	\end{figure}

\end{document}